\documentclass[onefignum,onetabnum]{siamart220329}
\usepackage{amsmath}
\usepackage{comment,cite}
\usepackage{hyperref}
\usepackage{cleveref}
\usepackage{cite}
\usepackage{enumitem}

\DeclareFontFamily{U}{mathx}{}
\DeclareFontShape{U}{mathx}{m}{n}{<-> mathx10}{}
\DeclareSymbolFont{mathx}{U}{mathx}{m}{n}
\DeclareMathAccent{\widehat}{0}{mathx}{"70}
\DeclareMathAccent{\widecheck}{0}{mathx}{"71}
\DeclareMathOperator{\range}{range}

\newcommand{\lat}{\mathcal{R}}
\newcommand{\scl}{\mathcal{R}^{s}}
\newcommand{\n}{n_{\text{poly}}}
\newcommand{\sci}{\mathcal{A}^{s}}

\newcommand{\ctalo}{\tau_{\alpha_1}}
\newcommand{\ctalt}{\tau_{\alpha_2}}
\newcommand{\stalo}{\tau_{\alpha^{s}_{1}}}
\newcommand{\stalt}{\tau_{\alpha^{s}_{2}}}

\newcommand{\osum}{\sum_{R\alpha\in \mathcal{I}_{2,d}}}

\newcommand{\cml}[1]{{{\color{red}{[ML: #1]}}}}

\usepackage{lipsum}

\usepackage{amsfonts}
\usepackage{graphicx}
\usepackage{epstopdf}
\usepackage{algorithmic}
\ifpdf
  \DeclareGraphicsExtensions{.eps,.pdf,.png,.jpg}
\else
  \DeclareGraphicsExtensions{.eps}
\fi

\usepackage{enumitem}
\setlist[enumerate]{leftmargin=.5in}
\setlist[itemize]{leftmargin=.5in}

\newsiamremark{remark}{Remark}
\newsiamremark{hypothesis}{Hypothesis}
\crefname{hypothesis}{Hypothesis}{Hypotheses}
\newsiamthm{claim}{Claim}

\headers{
    Learning the local density of states of a bilayer moiré material 
}{D. Liu, A. Watson, M. Hott,
S. Carr and M. Luskin}

\title{Learning the local density of states of a bilayer moiré material in one dimension
\thanks{Submitted to the editors \today.
\funding{ML's and MH's research was supported in part by Simons Targeted Grant Award No. 896630.   DL's, AW's, and ML's research was supported in part by NSF DMREF Award No. 1922165. ML's research was supported in part by a Visiting Scholarship at the Simons Flatiron Institute. AW's research was supported in part by NSF grant DMS-2406981.}}}

\author{Diyi Liu\thanks{School of Mathematics, University of Minnesota—Twin Cities, Minneapolis, Minnesota 55455, USA
  (\email{liu00994@umn.edu}).}
 \and Alexander B. Watson\thanks{School of Mathematics, University of Minnesota—Twin Cities, Minneapolis, Minnesota 55455, USA
  (\email{abwatson@umn.edu}).}
 \and Michael Hott\thanks{School of Mathematics, University of Minnesota—Twin Cities, Minneapolis, Minnesota 55455, USA
  (\email{mhott@umn.edu}).}
\and Stephen Carr\thanks{Department of Physics, Brown University, Providence, Rhode Island 02912-1843, USA 
  (\email{stcarr.nj@gmail.com}).}
\and Mitchell Luskin\thanks{School of Mathematics, University of Minnesota—Twin Cities, Minneapolis, Minnesota 55455, USA
  (\email{luskin@umn.edu}).}
}

\usepackage{amsopn}

\newtheorem*{problem*}{Inverse problem}

\begin{document}

\maketitle

\begin{abstract} Recent work of three of the authors showed that the operator which maps the local density of states of a one-dimensional untwisted bilayer material to the local density of states of the same bilayer material at non-zero twist, known as the twist operator, can be learned by a neural network. In this work, we first provide a mathematical formulation of that work, making the relevant models and operator learning problem precise. We then prove that the operator learning problem is well-posed for a family of one-dimensional models. To do this, we first prove existence and regularity of the twist operator by solving an inverse problem. We then invoke the universal approximation theorem for operators to prove existence of a neural network capable of approximating the twist operator.
\end{abstract}




\section{Motivation and summary} \label{sec:mot_and_sum}

Since the discovery of correlated insulating phases and superconductivity in twisted bilayer graphene in 2018 \cite{Cao2018a,cao2018unconventional}, it has become clear that stackings of 2D materials with a relative interlayer twist provide an excellent platform for investigating quantum many-body electronic phases. The huge number of novel materials which can be realized in this way motivates the development of computational tools for efficiently predicting which 2D materials stackings, and which twist angles, are likely to support such phases.

The stacking-dependent local density of states (SD-LDOS), the local density of states defined on the space of stackings (configurations) rather than in real space, provides a useful probe of whether a given twisted bilayer is likely to support interesting quantum many-body electronic phases. This is because, although the SD-LDOS is computed entirely from a single-particle model, peaks in the SD-LDOS correspond to energies and stackings where electron-electron interactions have the potential to be enhanced. Peaks in the SD-LDOS can be viewed as the generalization of flat Bloch bands of effective moir\'e-scale continuum models to materials where no moir\'e-scale continuum model is available \cite{duality20}.


The basic challenge in predicting electronic properties such as the SD-LDOS of a twisted bilayer material is their lack of periodicity at generic (incommensurate) twist angles. Although powerful computational methods which leverage ergodicity of the twisted bilayer have been developed\cite{massatt2017electronic,massatt2017incommensurate,genkubo17,twistronics}, it remains much easier to compute the properties of untwisted stackings, which retain the periodicity of the monolayer. This is especially true for density functional theory approaches, to which approaches based on ergodicity have not yet been generalized (although see \cite{ZHOU201999,wang2024convergence,CANCES2013241}).

These considerations lead to a natural question: is it possible to predict electronic properties such as the SD-LDOS of a twisted bilayer material directly from the electronic properties of the same material at zero twist? If this were possible, it would allow for the discovery of new twisted bilayer materials likely to support interesting many-body phases relying only on relatively simple computations on untwisted bilayers. The same question can be posed more mathematically: do there exist operators which map properties of models of untwisted bilayer materials to those of the same bilayer materials at non-zero twist? Could these operators be numerically computed? In what follows, we will refer to such operators as twist operators.

Recent work of three of the authors \cite{Liu_2022} suggests that the answer to these questions may be yes. They considered a dataset consisting of pairs of images: on one hand, an image showing the numerically computed SD-LDOS for an untwisted one-dimensional bilayer material, and on the other, the numerically computed SD-LDOS for the same one-dimensional bilayer material, but twisted. Here, ``twisting'' a one-dimensional bilayer material refers to slightly shortening the lattice constant of one of the layers to create a lattice mismatch analogous to twisting a 2D bilayer. They then showed that it is possible to train a neural network on such a dataset so that the network accurately predicts the twisted SD-LDOS image corresponding to untwisted SD-LDOS images not in the training dataset. In other words, not only does a twist operator for the SD-LDOS exist, it is possible for a neural network to effectively learn it.  Our model and analysis can also be applied to 2D untwisted heterostructures \cite{TMDHubbard} directly where the lattice shortening models lattice mismatch.

The present work has two main goals. The first is to provide a mathematical formulation of \cite{Liu_2022}, where we describe the one-dimensional models considered in \cite{Liu_2022} in detail, introduce the twist operator for the SD-LDOS in a general setting, and then make the specific operator learning problem considered in \cite{Liu_2022} precise. The second goal is to provide insight into the results of \cite{Liu_2022} by initiating the rigorous theory of the twist operator for the SD-LDOS. More specifically, for a specific family of one-dimensional bilayer models, and when the SD-LDOS is approximated by a polynomial, we are able to prove existence of the twist operator relating untwisted SD-LDOS images to their twisted couterparts. The key step is to prove that the coefficients of the untwisted bilayer Hamiltonian can be recovered from untwisted SD-LDOS data. Moreover, we are able to show that the twist operator is continuous, and hence can be approximated by a neural network up to arbitrary accuracy by invoking the universal approximation theorem.

\subsection{Related work} \label{sec:rel}


The learning problem considered here and in \cite{Liu_2022} is an operator learning problem. Other operator learning problems have been considered in the mathematical literature; see, for example, \cite{batlle2024kernel,li2020fourier,li2020neural,boulle2023mathematical,lu2021learning}. Note that the task of learning the solution $u$ of a differential equation $\mathcal{L} u = f$ can be interpreted as learning the action of the operator $\mathcal{L}^{-1}$ \cite{khoo2021solving,hsieh2018learning}. 

The problem of learning the unknown Hamiltonian operator for a given quantum system from data has been considered in the physics literature, although we emphasize that these works do not develop rigorous mathematical theory justifying their methods. For example, using band structure data or other kinds of electronic observable data \cite{wang2021machine,schattauer2022machine,gu2023deeptb,trilearn23,khosravian2024hamiltonian}, or atomistic local environment data \cite{li2022deep,gong2023general,hegde2017machine,zhang2022equivariant}. 
The problem of predicting atomic reconstruction in materials using machine learning was considered in \cite{aditya2023wrinkles}. 
The problem of identifying materials likely to host flat bands using machine learning has also been considered in the physics literature; for example, see \cite{tritsaris2021computational,bhattacharya2023deep}. 

Some related mathematical work has been done in the setting of quantum information science, where the works \cite{ni2023quantum,li2023heisenberg, huang2023learning,mirani2024learning} have introduced rigorously efficient algorithms for learning unknown parameters of many-body Hamiltonians from dynamics data. It would be interesting to try to apply similar ideas to the inverse problem we consider here, but note that the problem is quite different, since here we consider spectral (density of states) data.

\subsection{Structure of this work} \label{sec:outline}

The structure of this work is as follows. In Section \ref{sec:models}, we introduce the class of discrete tight-binding models of one-dimensional bilayer materials from which the dataset of \cite{Liu_2022} was constructed. In Section \ref{sec:LDOS}, we define the stacking-dependent local density of states (SD-LDOS) and the concept of SD-LDOS images. In Section \ref{sec:twist_op}, we introduce the concept of the twist operator for SD-LDOS images considered in \cite{Liu_2022} and present our results proving existence and regularity of the twist operator under specific conditions. In Section \ref{sec:approximation_by_NN}, we apply the universal approximation theorem for operators to prove existence of a neural network approximating the twist operator. In Section \ref{sec:commen}, we discuss computing SD-LDOS images for commensurate twists, which is how the data in \cite{Liu_2022} were generated. We provide our Conclusion in Section \ref{sec:conc}.

\section{Mathematical formulation of the coupled chain model} \label{sec:models}
Consider two periodic atomic chains in parallel within a certain plane. 
For the coupled chain, the Bravais lattice for each of the two chains can be defined as
\begin{align*}
    \mathcal{R}_{i}&=\left \{ na_i : n \in \mathbb{Z}  \right \}, \quad i\in \{1,2 \},
\end{align*}
where $a_i$ is the lattice parameter for the $i$th chain. We assume $a_1=1$ and $a_2=1-\theta$ in the paper where $\theta<1$ is the lattice mismatch for the coupled chain. The atomic orbitals attached to the two chains are labelled by the indices,
\begin{align*}
    \Omega= (\mathcal{R}_1\times \mathcal{A}_1) \cup (\mathcal{R}_2\times \mathcal{A}_2),
\end{align*}
where $\mathcal{A}_i$ denotes the set of indices of atomic positions within each unit cell of chain $i$ and $m_i:=|\mathcal{A}_i|.$ Each atomic orbital is attached to an atom, and each atom is indexed by a unit cell and the shift in the
unit cell.  We define the unit cell for $i$-th chain as
\begin{align*}
    \Gamma_i&=\left \{  \beta a_i : \beta \in [0,1)  \right \}, \quad i\in \{1,2 \},
\end{align*}
where each atomic orbital labeled by $R_i\alpha_i \in \Omega$ for $R_i \in \mathcal{R}_i$ and $\alpha_i \in \mathcal{A}_i$ is located at $R_i+\tau_{\alpha_i}$ for shifts
$\tau_{\alpha_1} \in [0,1]$ and $  \tau_{\alpha_2} \in [0,1-\theta]$. 

\begin{figure}[ht]
    \centering\includegraphics[width=0.8\linewidth]{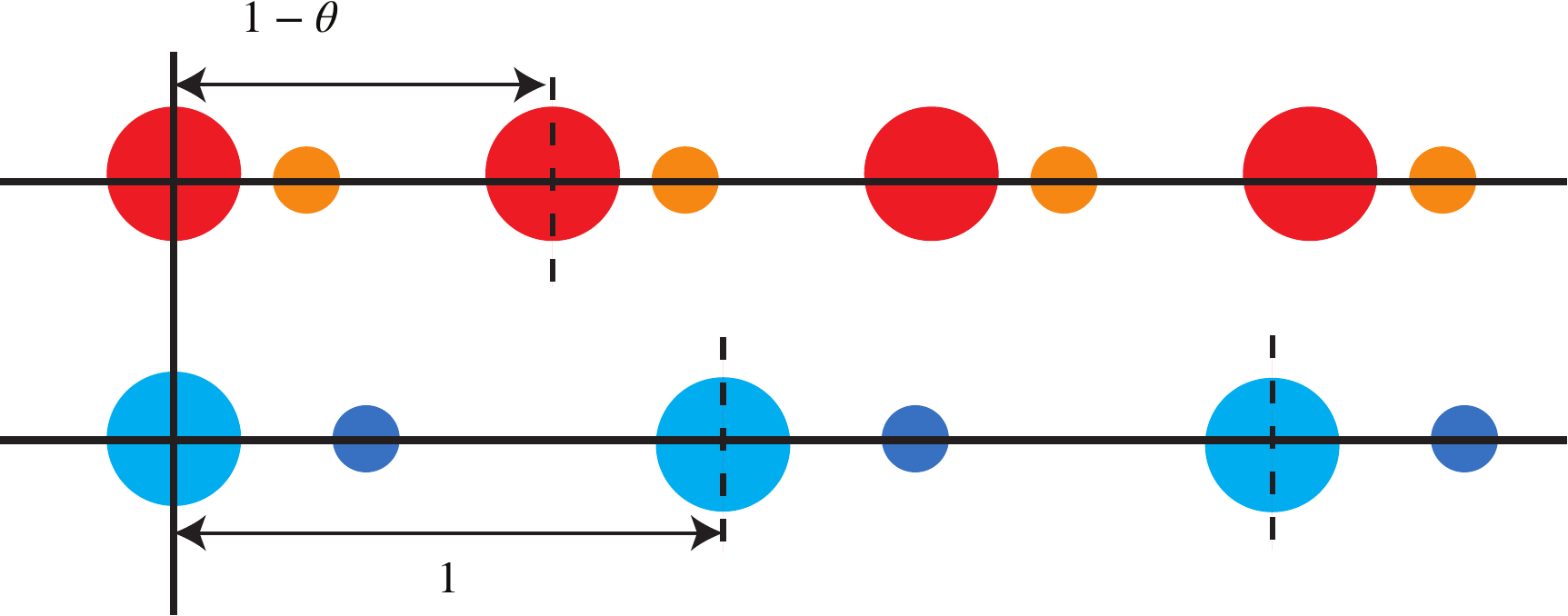}
    \caption{Schematic of the coupled chain system described in Section \ref{sec:models} with two orbitals per cell ($m_1 = m_2 = 2$), with a positive lattice mismatch $\theta > 0$ without interlayer shift ($d = 0$). When $\theta$ is irrational, the system has no exact periodic cell.}
    \label{fig:incom}
\end{figure}
The general coupled chain system is aperiodic when $\theta$ is irrational, but the emergence of a new scale called the \textit{moir\'e scale} can be observed, whose length is given by
$$a_{M}:=\left(\frac{1}{a_1}-\frac{1}{a_2}\right)^{-1}=\frac{1-\theta}{\theta}.$$ 
The disregistry function is defined by $$b_{1\to 2}(x):=(1-a_2a_1^{-1}) x      \mod \mathcal{R}_2,$$
where for any $x\in \mathbb{R}$, $x\mod\mathcal{R}_2 = \left \{ x-na_2 \ | \  n\in \mathbb{Z} \ \text{such that} \ 0\leq x-na_2<a_2 \right \}$. The disregistry function describes the local environment for an atomic orbital located at $x$. The moir\'e period 
is defined to be the period of the disregistry function, i.e., \  $b_{1\to 2}(a_M)=0$ \cite{massatt2017electronic, massatt2017incommensurate, relaxfosdick22, relaxphysics18, cazeaux2018energy, watson2022bistritzer}.

\subsection{The tight binding model}
In the tight-binding model \cite{massatt2017electronic, kaxiras_joannopoulos_2019,massatt2017incommensurate}, we model the wave function of an electron as an element $\psi$ of the Hilbert space
\begin{equation*}
     \ell^2(\Omega) = \left\{ (\psi(R_i\alpha_i))_{R_i\alpha \in \Omega} : \sum_{R_i\alpha_i \in \Omega} |\psi(R_i\alpha_i)|^2 < \infty \right\}.
\end{equation*}
The tight binding Hamiltonian $H$ is defined as the linear operator on $\ell^2(\Omega)$ acting as
\begin{equation} \label{eq:TB_H}
  \left(H\psi\right)(R\alpha)= \sum_{R'\,\alpha'\in \Omega} H(R\alpha,R'\alpha')  \psi(R'\alpha'),
\end{equation}
where $H(R\alpha,R'\alpha')$ is referred as the ``hopping energy,'' which models the interaction of the orbital associated with each atom. When $R\alpha, R'\alpha' \in \lat_i \times \mathcal{A}_i$, $H(R\alpha,R'\alpha')$ is referred as an intralayer hopping energy. Moreover, when $R\alpha \in \lat_i \times \mathcal{A}_i$ and $R'\alpha' \in \lat_j \times \mathcal{A}_j$ but $i\neq j$,  $H(R\alpha,R'\alpha')$ is referred as an interlayer hopping energy. 

The intralayer hopping energy for the coupled chain where $R \alpha, R'\alpha'\in \mathcal{R}_i\times \mathcal{A}_i$ can be modeled as
\begin{equation}\label{eq:hop0}
    H(R\alpha,R'\alpha') = \begin{cases}
       t_{\alpha,\alpha'},\, & R-R'=a_i, \\
       t_{\alpha',\alpha},\, & R-R'=-a_i ,\\
       \epsilon_{\alpha,\alpha'}, \, &R=R',
       \\ 0, \, &\text{else}.
    \end{cases}
\end{equation}
The constants $\epsilon_{\alpha,\alpha'}$ and $t_{\alpha,\alpha'}$ are given in $\mathbb{R}$ for a specific material system and atomical orbital index $\alpha,\alpha'$.  The interlayer hopping energy will be modeled by a smooth interlayer hopping function $h : \mathbb{R} \rightarrow \mathbb{R}$ such that 
\begin{equation}
\label{eq:hop1}
    H(R\alpha,R'\alpha')=h(|r|), \quad r=R+\tau_{\alpha}+d(-\delta_{i1}+\delta_{i2})-R'-\tau_{\alpha'},
\end{equation}
where $R \alpha \in \mathcal{R}_i\times \mathcal{A}_i$,  $R'\alpha'\in \mathcal{R}_j\times \mathcal{A}_j$ and $i\neq j$. The parameter $d$ represents a relative shift between the layers, illustrated in \cref{fig:incom2}. At incommensurate twist angles, averages of physical quantities over $R$, such as the density of states, can be efficiently represented as averages over local configurations $d$ using ergodicity \cite{massatt2017electronic,genkubo17}. We assume that $h$ decays exponentially in the sense that there exist positive constants $h_0, \gamma > 0$ such that 
\begin{equation}
\label{equ:h_cond}
    |h(x)| \leq h_0 e^{-\gamma |x|}.
\end{equation}

\begin{figure}[ht]
    \centering
    \includegraphics[width=0.8\linewidth]{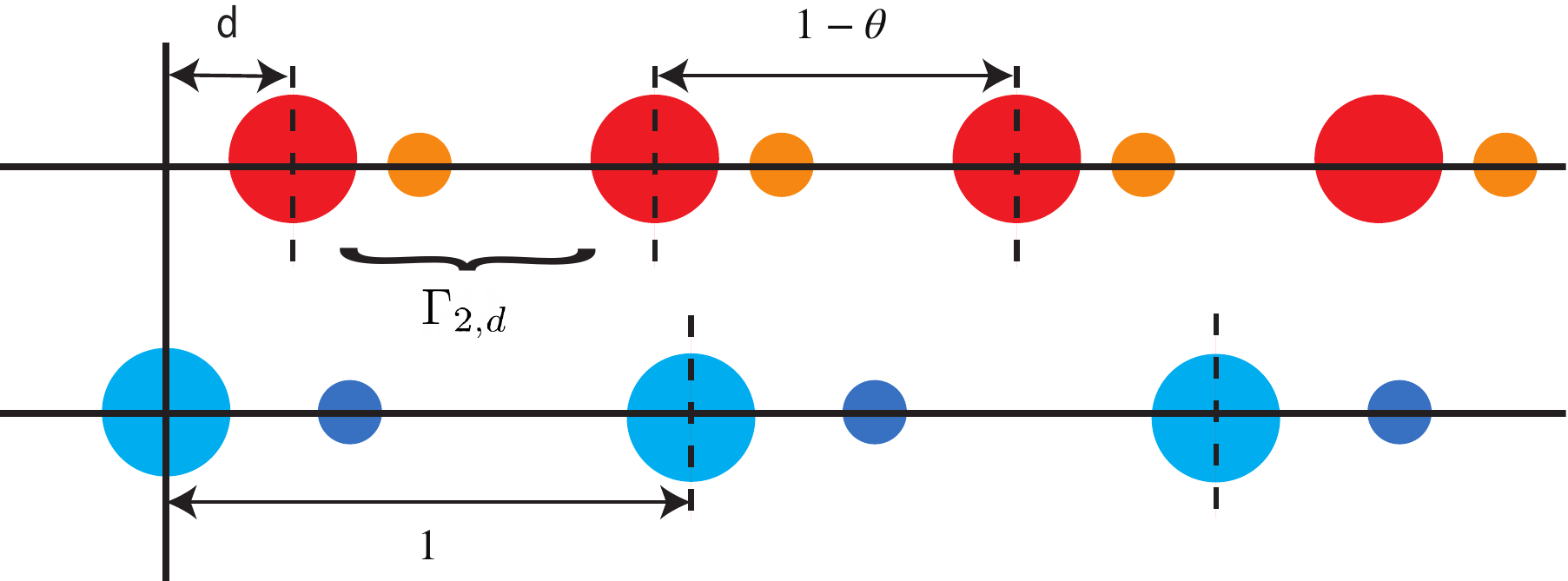}
    \caption{   Schematic of the coupled chain system described in Section \ref{sec:models} with two orbitals per cell ($m_1 = m_2 = 2$), with a positive lattice mismatch $\theta > 0$ with interlayer shift ($d \neq 0$).}
    \label{fig:incom2}
\end{figure}

Given the assumption on interlayer coupling function and intralayer coupling function, by Young's inequality \cite{folland1999real}, the tight binding Hamiltonian is a bounded operator on $\ell^{2}(\Omega)$; see, e.g., \cite{kong2023modeling}. To emphasize the dependence of $H$ on these parameters, we will also use $H_{\theta,d}$ to denote the tight binding Hamiltonian for the coupled chain with mismatch $\theta$ and shift $d$.

\section{Stacking-dependent local density of states: analysis and computation} 
\label{sec:LDOS}
The local density of states (LDOS) measures the density of eigenvalues and associated (generalized) eigenvectors of the Hamiltonian as a function of energy (eigenvalue) and position. It can be probed experimentally and is important for understanding the electronic properties of solids. The emergence of a singularity in the LDOS of twisted layered material, an indicator of a highly correlated physical system, give rise to the possibility of superconductivity and other interesting electronic properties \cite{Liu_2022}. The local density of states can also be used for computing the density of states, another important electronic property of materials \cite{massatt2017electronic,genkubo17}. Predicting the LDOS of a solid is a crucial task for condensed matter physics, quantum chemistry as well as quantum material science.

\subsection{Stacking-dependent local density of states}

We introduce notation for one distinguished, shifted unit cell in layer two
\begin{equation} \label{eq:home_cell}
    \Gamma_{2,d} :=\left \{ d+(1-\theta) \beta: \beta \in [0,1)  \right \},
\end{equation}
and for the orbitals within that cell
\begin{align*}
    \mathcal{I}_{2,d} := \{ R\alpha \in \lat_2\times \mathcal{A}_2 \ | \  R+\tau_{\alpha}+d \in \Gamma_{2,d}  \}.
\end{align*}

We can then introduce the smoothed stacking-dependent local density of states as
\begin{equation}
\label{equ:ldos}
        \rho_{\theta}(d,E):=\osum g( H_{\theta,d }-E)(R\alpha,R\alpha),
\end{equation}
where $g$ is a positive even polynomial approximating the Dirac delta function. We can choose $g$ to be an approximation of the normal distribution with variance inversely proportional to the square of the polynomial degree\cite{Weisse2006,massatt2017electronic}. The use of polynomial approximations to the delta function, constructed using Chebyshev polynomials, to compute the local density of states is known as the kernel polynomial method \cite{Weisse2006}. Note that we only consider the local density of states in a single cell. This is because the local density of states in other cells can be obtained by different choices of the shift $d$ \cite{massatt2017electronic}.


The degree of the polynomial $g$ will play a significant role in what follows and we denote this degree by $n_{\text{poly}}$. Since $g$ is a polynomial and $H_{\theta,d}$ depends smoothly on $\theta$ and $d$, the smoothed stacking-dependent local density of states is smooth as a function of $\theta, d$, and $E$.


\subsection{The local density of states image}

In practice, one generally only has access to the local density of states evaluated on a discrete grid of $d$ and $E$ values, sampled from the monolayer unit cell and an energy window of interest, respectively. This is because one can only perform finitely many experiments or computations. To model this, we introduce the concept of the local density of states image, defined as follows. We first introduce uniform grids of $d$ and $E$ values for $N_d, N_E \geq 2$
\begin{gather}
    \mathcal{X}:= \left \{ d_1, d_2, d_3 \dots d_{N_d} \right  \}, \quad d_i = \frac{(i-1) a_1}{N_d}, \quad 1 \leq i \leq N_d, \label{eq:d_vals} \\ 
    \mathcal{Y}:= \left \{ E_1, E_2, E_3 \dots E_{N_E} \right  \}, \quad E_j = E_1 + \frac{(j-1) (E_{N_E} - E_1)}{N_E}, \quad 1 \leq j \leq N_E, \label{eq:E_vals}
\end{gather}
where $E_1 < E_{N_E}$ are fixed. For example, we could choose $E_1$ and $E_{N_E}$ such that they bound the whole spectrum of $H_{\theta,d}$ for all values of $d$, or we could choose them to bound one important subset of the spectrum, such as energies near to the Fermi level \cite{ashcroft_mermin}. We refer to the set $\mathcal{X} \times \mathcal{Y}$ as the LDOS grid, and we define the local density of states image $\varrho_{\theta,\mathcal{X},\mathcal{Y}}$ as the smoothed stacking-dependent local density of states restricted to the LDOS grid $\mathcal{X} \times \mathcal{Y}$, i.e.,
\begin{equation} \label{eq:LDOSgrid}
    \varrho_{\theta,\mathcal{X},\mathcal{Y}} := \left( \varrho_{\theta,i,j} \right)_{1 \leq i \leq N_d, 1 \leq j \leq N_E} \in \mathbb{R}^{N_d \times N_E}, \quad \varrho_{\theta,i,j} := \rho_{\theta}(d_i,E_j), \quad (d_i,E_j) \in \mathcal{X} \times \mathcal{Y}.
\end{equation}
Note that the LDOS grid is fully specified by the choice of the LDOS grid parameters $N_d, E_1, E_{N_E},$ and $N_E$. Typical LDOS images can be seen in Figure \ref{fig:diagram}.

 
\section{Introduction to the twist operator}  \label{sec:twist_op}

The stacking-dependent local density of states can be computed efficiently when $\theta = 0$ for all values of $d$ and $E$, because, in this case, the system is periodic at the scale of the monolayer lattice constant. The computation is more difficult for $\theta \neq 0$ because periodicity at the lattice scale is generally broken. 

The central idea of the work \cite{Liu_2022} is that there may exist an operator, dubbed the \emph{twist operator}, mapping local density of states values calculated at $\theta = 0$ to their associated values at $\theta \neq 0$. In other words, the idea is that there may exist an operator
\begin{equation}
\begin{split} \mathcal{L}_{\theta } : \ C(\Gamma_2\times \mathbb{R}) &\rightarrow C(\Gamma_2\times \mathbb{R}), \\
 \rho_{0}(d,E)&\mapsto \rho_{\theta}(d,E),
\end{split}
\end{equation}
where $\rho_\theta$ is as in \eqref{equ:ldos}. More modestly, we may hope that, for a fixed LDOS grid $\mathcal{X} \times \mathcal{Y}$, there exists a discrete twist operator acting on local density of states images
\begin{equation}
\label{eq:numldos}
\begin{split} \Tilde{\mathcal{L}}_{\theta } : \  \mathbb{R}^{N_d\times N_E}&\rightarrow \mathbb{R}^{N_d\times N_E} ,\\
    \varrho_{0,\mathcal{X},\mathcal{Y}}&\mapsto \varrho_{\theta,\mathcal{X},\mathcal{Y}}.
\end{split}
\end{equation}
The work \cite{Liu_2022} provided evidence for the existence of such an operator \eqref{eq:numldos} by successfully training a neural network to approximate this map. The neural network was trained using a database of LDOS image pairs, each pair consisting of one LDOS image computed at $\theta = 0$, called the input image, and a second LDOS image computed with the same Hamiltonian with $\theta \neq 0$, called the output image. 

\subsection{Existence of the twist operator: an inverse problem perspective}

The work \cite{Liu_2022} also suggested a route to rigorously proving existence of such a map, i.e., to proving well-posedness of the learning problem addressed in \cite{Liu_2022}. The idea is that the twist operator could be defined as the composition of two maps. First, the inverse map from the LDOS function (image) at $\theta = 0$ to the set of parameters defining the tight-binding model. Then, the forward map from the set of parameters defining the tight-binding model to the LDOS function (image) at $\theta \neq 0$. Since the forward map is well-defined through \eqref{equ:ldos}, \emph{a sufficient condition for the twist operator to exist is existence of the inverse map from LDOS functions (images) to tight-binding parameters.} This idea is illustrated in Figure \ref{fig:diagram}. Another important consideration is the regularity of the twist operator map between LDOS functions (images). In particular, for the map to be rigorously learnable by a neural network, it must be continuous. Since the forward map \eqref{equ:ldos} is smooth, a sufficient condition for this to hold is continuity of the inverse map from LDOS functions (images) to model parameters.

In the following sections, we will make these ideas more precise (Section \ref{sec:twist_theory}), and then present two theorems establishing sufficient conditions for existence and continuity of the twist operator between LDOS images \eqref{eq:numldos} (Section \ref{sec:inverse}). We will then discuss theorems guaranteeing that the operator can be learned by a neural network (Section \ref{sec:approximation_by_NN}). 

\begin{figure}
    \centering
    \label{fig:theory}
\includegraphics[width=0.7\linewidth]{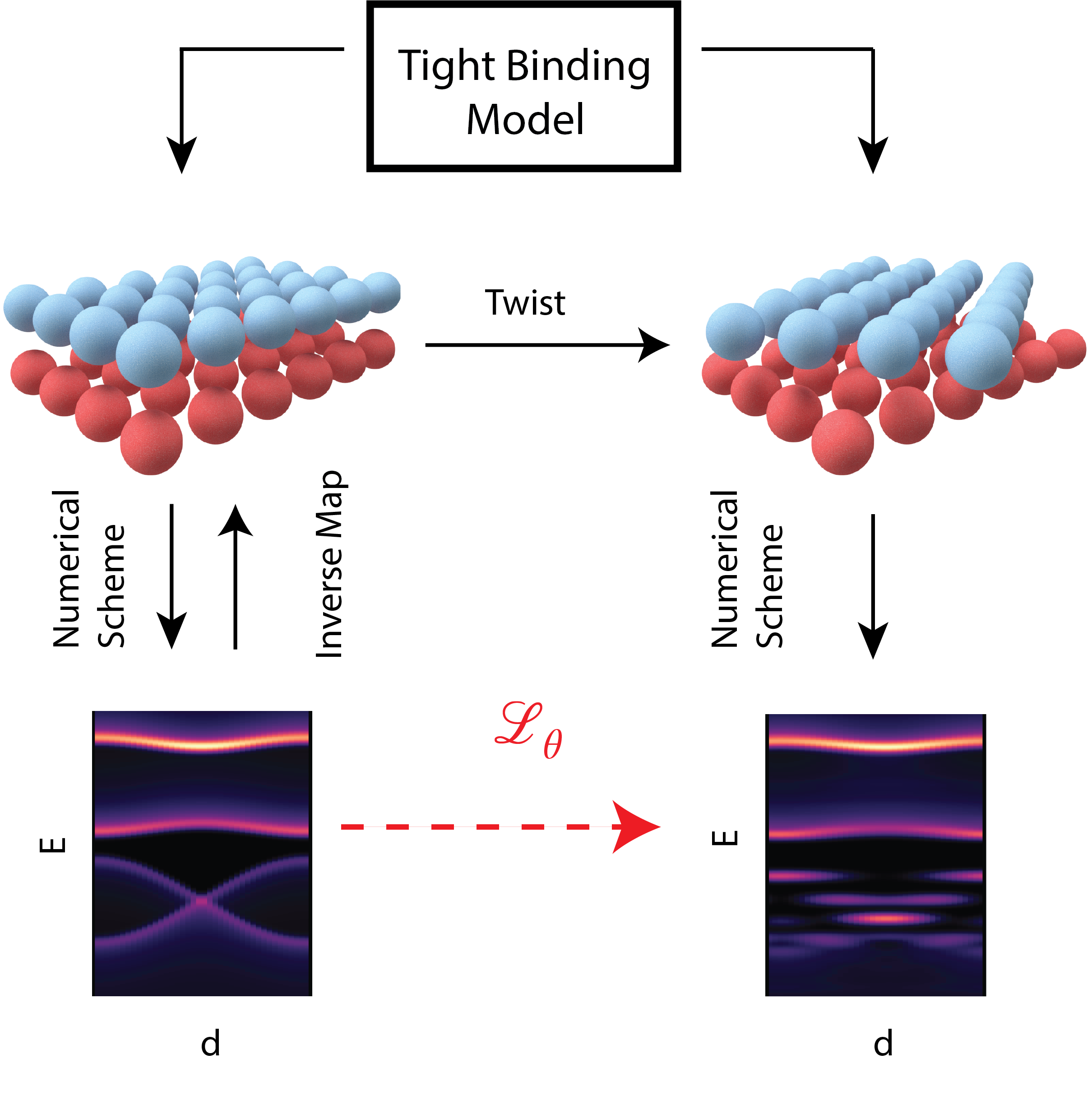}
    \caption{Schematic showing relationships between tight-binding models of untwisted and twisted structures and stacking-dependent LDOS images. The twist operator $\mathcal{L}_\theta$ learned in \cite{Liu_2022} maps from untwisted LDOS images to twisted LDOS images. It can be defined via the inverse map from untwisted LDOS images to tight-binding model parameters, composed with the forward map from these tight-binding model parameters to the LDOS image of the twisted bilayer tight-binding model. 
    }
    \label{fig:diagram}
\end{figure}

\subsection{Rigorous theory of the twist operator: precise inverse problem}
\label{sec:twist_theory}

We will establish existence and continuity of the twist operator for a more specific family of tight-binding models than given in Section \ref{sec:models}. We will discuss the generalization of the theory for other tight binding models in Remark~\ref{rmk:general}. We start by assuming that both chains have only one orbital per unit cell. The Hilbert space is then
\begin{equation} \label{eq:L2_again}
    \ell^2(\Omega) = \left\{ (\psi(R_i))_{R_i \in \Omega} : \sum_{R_i \in \Omega} |\psi(R_i)|^2 < \infty \right\}, \quad \Omega = \mathcal{R}_1 \cup \mathcal{R}_2.
\end{equation}
The tight-binding Hamiltonian $H$ is again given by \eqref{eq:TB_H}, where, for $R, R' \in \mathcal{R}_i$, i.e., for the intralayer part of the Hamiltonian, we have
\begin{equation}
    H(R,R') = \begin{cases} 
        t,\, & R-R'=a_i, \\
       t,\, & R-R'=-a_i, \\
       \epsilon, \, &R=R',
       \\ 0, \, &\text{else}.
    \end{cases}
\end{equation}
For $R \in \mathcal{R}_i, R' \in \mathcal{R}_j$ with $i \neq j$, i.e., for the interlayer part of the Hamiltonian, we have
\begin{equation} \label{eq:interlayer}
    H(R,R') = h(|r|), \quad r = R + d(-\delta_{i1}+\delta_{i2}) - R',
\end{equation}
where we assume the interlayer hopping function has the form
\begin{equation} \label{eq:inter}
    h(|r|) = \nu f\left(\frac{r}{l}\right)
\end{equation}
for some fixed decaying function $f : \mathbb{R} \rightarrow \mathbb{R}$. Once $f$ is specified, we see that the model depends on precisely five model parameters: $\epsilon, t, \nu, l,$ and $\theta$ (recall that the model depends on $\theta$ since the lattice constant of $\mathcal{R}_2$ is $1 - \theta$, which enters the interlayer hopping through \eqref{eq:interlayer}). At this point, we will assume that while $\epsilon$ is an arbitrary real number, the parameters $t, \nu$, and $l$ are all positive, i.e., we assume $(\epsilon, t, \nu, l) \in \mathbb{R} \times \mathbb{R}^3_+$, where $\mathbb{R}_+$ denotes the positive real half line. For any choice of model parameters and LDOS grid parameters $N_d, E_1, E_{N_E}$ and $N_E$, we can define LDOS images through \eqref{equ:ldos} and \eqref{eq:LDOSgrid}.

We can now pose the following precise inverse problem regarding the model when $\theta = 0$.
\begin{problem*}
    Suppose we are given an LDOS image on the LDOS grid with parameters $N_d, E_1, E_{N_E}$ and $N_E$ generated from the model \eqref{eq:L2_again}-\eqref{eq:inter} with some (unknown) choice of parameters $\epsilon, t, \nu$, and $l$, with $\theta = 0$, and $f$ fixed (and known). Does the LDOS image uniquely specify the model parameters $\epsilon, t, \nu$, and $l$? In other words, is the map from parameters $\epsilon, t, \nu$, and $l$ to LDOS images one-to-one? When the map is one-to-one so that the inverse map from LDOS images to model parameters exists, is the inverse map continuous?   
\end{problem*}

\begin{remark}
\label{rmk:general}
    In this work, we consider the inverse problem only for the specific tight-binding model introduced in this section. It is natural to ask whether our methods generalize to more general tight-binding models. For example, the intralayer coupling of two orbitals within two cells may not be vanishing if the distance between the two cells is larger than lattice parameter. There could be more than one orbital in each unit cell or for each atomic site. We could also consider systems in two dimensions and higher. For the first two cases, the analogous inverse problem would require solving a more complicated polynomial system and may require larger $N_E$ to establish a one-to-one mapping. For the last case, the dimensionality of the lattice will influence how many $d$-grid points are needed to make the inverse problem well-posed, but we expect that with just one orbital per unit cell this case is more tractable than the other two. Since the necessary calculations become much more complicated in these cases, we leave them to future works.
\end{remark}

In the following section, we will identify assumptions under which this inverse problem can be solved, and hence existence of the twist operator proved.

\subsection{Rigorous theory of the twist operator: solving the inverse problem}  \label{sec:inverse}

Our first theorem establishes that the inverse problem can be solved for a specific choice of $N_d$, sufficiently high resolution in energy $N_E$, when the interlayer hopping function $f(\cdot) = e^{- |\cdot|}$, assuming that the parameters are drawn from a compact subset of $\mathbb{R} \times \mathbb{R}_+^3$.
\begin{theorem}{\textbf{Existence and Stability of Inverse Map}}
\label{thm:in_specific}
    Assume $N_E > n_{\text{poly}}$, where $n_{\text{poly}}$ is the degree of the polynomial $g$ in \eqref{equ:ldos}. Let $E_{N_E} > E_{1}$ be otherwise arbitrary. Let $N_d \geq 3$, and assume that the points $0,\frac{1}{4},\frac{1}{2}$ are included in the $d$ grid $\mathcal{X}$. Thus, the LDOS grid is given by $\mathcal{X} \times \mathcal{Y}$, where
    \begin{equation}
        \left\{0,\frac{1}{4},\frac{1}{2}\right\} \subset \mathcal{X}\quad\text{and} \quad \mathcal{Y} = \left\{ E_1, ... , E_{N_E} \right\}.
    \end{equation}
    Assume further that the interlayer coupling function is given by \eqref{eq:inter}, with $f$ exponential, so that
\begin{equation*}
    h(|r|)=\nu e^{-\frac{|r|}{l}}.
\end{equation*}
    Finally, we assume that $(\epsilon, t, \nu,l)$ are constrained to a compact subset $\mathcal{P}_{\text{comp}} \subset \mathbb{R} \times \mathbb{R}_+^3$. In particular, we assume $t, \nu, l$ are strictly positive. Then, the mapping 
\begin{equation}
\begin{split}
        \mathcal{K}_1: 
        \mathcal{P}_{\text{comp}} &\rightarrow \mathbb{R}^{N_d\times N_E}, 
        \\ (\epsilon, t,\nu,l) & \mapsto \varrho_{0,\mathcal{X},\mathcal{Y}},
\end{split}
\end{equation}
    from model parameters to LDOS images, defined for each value of the model parameters by \eqref{equ:ldos}-\eqref{eq:LDOSgrid}, is one-to-one, and hence a bijection onto its range. We also have the following stability result. In the special case where $E_i + E_{N_E-i}=0$ for all $0\leq i \leq N_E$ and $N_E$ is odd, then
    \begin{equation} \label{eq:K_1_inv}
        \begin{split}|\mathcal{K}_1^{-1}|_{\infty} &\leq  \frac{C_{\mathcal{P}}}{2}\prod_{j=0}^{n_{poly}} \left( \sum_{i=0}^{n_{poly}-j} |a_{n_{poly}-i} C_{n_{poly}-i}^{n_{poly}-j-i}| \right) \\
        & \times \max_{\substack{i: E_i \geq 0}}\left \{ \left(1+\frac{1}{E_i} \right)\prod_{ \substack{j\neq i \\  E_j\geq 0}} \frac{1+E_j^2}{|E_i^2-E_j^2|}\right \}
        \end{split}
    \end{equation}
    where $C_{\mathcal{P}}$ is a nonzero constant depending on $\mathcal{P}_{comp}$ and $C_m^n$ is used to denote $ \frac{m!}{n! (m-n)!}$ for $ 0\leq n \leq m$. Note that since we can always shift the Hamiltonian by a constant so that its spectrum is centered at $0$, the condition $E_i + E_{N_E - i} = 0$ is not a significant restriction.
\end{theorem}
\begin{proof}
See proof in \cref{sec:in_specific}.
\end{proof}

\begin{remark}
We observe that the proof of the theorem is not restricted to the specific \( d \)-grid points \( 0, \frac{1}{4}, \frac{1}{2} \). In particular, the theorem holds for any set of \( d \)-points \( d_1, d_2, d_3 \), provided that the choice of these points ensures the values  
$$ S_{t,\nu,l}(d_1), S_{t,\nu,l}(d_2), S_{t,\nu,l}(d_3), $$ 
where $S_{t,\nu,l}(d)$ is defined in \eqref{equ:S},
uniquely determine the parameters \( t, \nu, l \). In particular, if the \( d \)-points satisfy the condition that  
\begin{equation}
\frac{e^{\frac{2d_2-2}{l}} + e^{\frac{-2d_2}{l}} - e^{\frac{2d_1-2}{l}} - e^{\frac{-2d_1}{l}}}{e^{\frac{2d_3-2}{l}} + e^{\frac{-2d_3}{l}} - e^{\frac{2d_2-2}{l}} - e^{\frac{-2d_2}{l}}}
\end{equation}
is a well-defined and strictly monotonic function with respect to \( l > 0 \), the proof remains valid for the new set of \( d \)-points. However, it is clear from how the parameters $t, \nu, l$ are recovered from the values of $S_{t,\nu,l}$ that the recovery of the parameters can be unstable if $d_1, d_2, d_3$ are chosen too close to each other. In particular, the constant $C_\mathcal{P}$ in \eqref{eq:K_1_inv} will blow up in this case. For more detail, see the discussion at the end of Appendix \ref{sec:in_specific}. 
\end{remark}

\begin{remark}
\label{rmk:invesesta}
    The bound on the norm $|\mathcal{K}_1^{-1}|_{\infty}$ \eqref{eq:K_1_inv} measures the stability of the map from LDOS images to the tight binding model. For the LDOS of a certain energy range, if the energy grid $\mathcal{Y}$ is fine enough, the problem becomes unstable. When $\Delta E$ is fixed, the upper bound can grow exponentially with respect to $n_{poly}$, i.e. for very sharply peaked approximations to the Dirac delta function. We note that the condition number of $\mathcal{K}_1^{-1}$ is not included in the result, which is necessary for the relative error analysis. The reason is that the relative error analysis is equivalent to the analysis of $
    |\mathcal{K}_1^{-1}|_{\infty}$ under the assumption that the tight binding model parameters are within a compact subset away from origin.
\end{remark}

The proof of Theorem \ref{thm:in_specific} relies on the fact that the local density of states \eqref{equ:ldos} has finite degree with respect to energy $E$. This means that the coefficient functions of the polynomial in $E$ can be recovered when we have LDOS values for greater than $n_{\text{poly}}$ different values of $E$. The tight binding parameters can then be recovered because of the special form of the interlayer coupling function. 

Our second theorem establishes that the inverse problem can be solved under different assumptions. We allow for a quite general class of interlayer coupling functions under the assumption that the parameter space is discretized and finite.
\begin{theorem}
\label{thm:in_general}
    Suppose that the parameter space is discrete and finite, i.e., we assume that the model parameters $\epsilon, t, \nu, l$ are drawn from a known finite subset $\mathcal{P}_{\text{fin}} \subset \mathbb{R} \times \mathbb{R}^3_+$. 
    Then, suppose that the interlayer hopping function $h(r):=\nu f(\frac{r}{l})\chi_{[-r_0,r_0+1)}$, where $\chi_I$ is the characteristic function for the interval $I$,
    for some fixed integer $r_0 > 1$, where $f$ is an 
    analytic function on a domain $\mathcal{D}$ with $\mathbb{R} \subset  \mathcal{D}\subseteq \mathbb{C}$, positive when restricted to the real line, and satisfying 
    \begin{equation} \label{eq:f_cond}
    \lim_{r\to \infty} \frac{f(\frac{r}{l})}{f(\frac{r}{l'}) } = 0, \quad \text{for any \ $0<l < l'$}.
\end{equation}
    Suppose that $N_E>n_{\text{poly}}$ and that $E_{N_E} > E_1$ but are otherwise arbitrary. Then, there exists a positive integer $N_0$ such that for any configuration space resolution $N_d>N_0$, the mapping 
\begin{equation}
\begin{split}
        \mathcal{K}_2: 
        \mathcal{P}_{\text{fin}} &\rightarrow \mathbb{R}^{N_d\times N_E}, 
        \\ (\epsilon, t,\nu,l) & \mapsto \varrho_{0,\mathcal{X},\mathcal{Y}},
\end{split}
\end{equation}
is one-to-one and hence a bijection onto its range.
\label{thm:inverseg2}
\end{theorem}
\begin{proof}
See proof in \cref{sec:in_general}.
\end{proof}

\begin{remark}
We assume the parameter space is discrete and finite in Theorem~\ref{thm:inverseg2}, which looks less general compared to the continuous parameter space assumption in Theorem~\ref{thm:in_specific}. However, we are compensated by the weaker assumption on the interlayer coupling function in Theorem~\ref{thm:inverseg2}. Moreover, this doesn't mean that the assumption on the discrete parameter set is not practical since the size of dataset for 2D material in certain environments is finite. The discrete condition is necessary because it guarantees that a positive and finite $c_0$ exist in \cref{equ:c0}. For any discrete parameter space or continuum parameter space, the theorem still holds if a positive and finite $c_0$ exists for the parameter space and the interlayer coupling function $f$. 
\end{remark}
\begin{remark}
Compared with Theorem~\ref{thm:in_specific}, the number of $d$-points needed is not explicitly given in statement of Theorem~\ref{thm:inverseg2}. However, the proof illustrates that for uniform $d$-grid points, the number of $d$ points $N_d>\max\{ 1,\frac{3c_1}{c_0}\}$ is enough for the existence of a one-to-one mapping. Here $c_0$ and $c_1$ are the two constants defined in \cref{equ:c0} and \cref{equ:c1}. Note that $c_0$ could become small in the continuum limit of the parameter space. The fact that $N_d$ blows up in this case reflects the necessity of the discreteness of the parameter space for our proof of Theorem \ref{thm:in_general} to go through.
\end{remark}

\begin{remark}\label{rem:stab_analytic}
Note that we do not carry out a detailed stability analysis for the map $\mathcal{K}_2$ as we did for $\mathcal{K}_1$ in Theorem \ref{thm:in_specific}. There are two potential sources of instability: in solving for moments of the Hamiltonian from LDOS images, and then in solving for the model parameters from moments of the Hamiltonian. In the first case, the analysis would be exactly the same as in Theorem \ref{thm:in_specific}, so we omit this here. In the second case, it is hard to carry out a stability analysis because we obtain existence of the mapping from moments of the Hamiltonian to model parameters by an argument by contradiction relying on analyticity of the interlayer hopping function $f$ (see Appendix \ref{sec:in_general}). It would be interesting to see whether alternative assumptions could be made here which would allow for such a stability analysis.
\end{remark}

The assumptions of the theorem are satisfied, for example, if $f(\cdot) = e^{- |\cdot|^2}$, and $\mathcal{P}_{\text{fin}}$ is a bounded uniform (equally spaced) grid, but uniformity is not necessary.

The proof of Theorem \ref{thm:in_general} is along the same lines as that of Theorem \ref{thm:in_specific}, except that one-to-oneness is proved by contradiction using the complex-analytic properties of $f$ and \eqref{eq:f_cond}.

With existence of the inverse map established either through Theorem \ref{thm:in_specific} or Theorem \ref{thm:in_general}, continuity of the inverse map is guaranteed by the following. Note that the result is trivial in the case of $\mathcal{K}_2^{-1}$ because the domain of $\mathcal{K}_2$ is discrete. 
\begin{theorem}
\label{thm:invcon}
    Let $\mathcal{K}_1$ and $\mathcal{K}_2$ be as in Theorem \ref{thm:in_specific} and Theorem \ref{thm:in_general}, respectively. Then, the inverse mappings $\mathcal{K}_1^{-1} : \range(\mathcal{K}_1) \rightarrow \mathcal{P}_{\text{comp}}$ and $\mathcal{K}_2^{-1} : \range(\mathcal{K}_2) \rightarrow \mathcal{P}_{\text{fin}}$ are continuous.
\end{theorem}
\begin{proof}
See proof in \cref{sec:invcon}.
\end{proof}

As discussed above, existence and continuity of the twist operator on LDOS images \eqref{eq:numldos} now follows. We summarize this result in the following.
\begin{theorem}
\label{thm:twist_op_con}
    Assume the conditions either of Theorem \ref{thm:in_specific} or of Theorem \ref{thm:in_general}, and let $\mathcal{K}$ be defined by $\mathcal{K}_1$ or $\mathcal{K}_2$ depending on the case. Note that we always have $\range(\mathcal{K}) \subset \mathbb{R}^{N_d \times N_E}$. Then, for any fixed $0 \leq \theta < 1$, there exists an operator
\begin{equation}
\begin{split}
    \tilde{\mathcal{L}}_\theta: \range(\mathcal{K}) &\rightarrow \mathbb{R}^{N_d\times N_E}, 
        \\ \varrho_{0,\mathcal{X},\mathcal{Y}} &\mapsto \varrho_{\theta,\mathcal{X},\mathcal{Y}}.
\end{split}
\end{equation}
Furthermore, this operator is continuous.
\end{theorem}
\begin{proof}
Theorems \ref{thm:in_specific} and \ref{thm:in_general} guarantee that we can define $\tilde{\mathcal{L}}_\theta := \mathcal{T}_\theta \mathcal{K}^{-1}$, where $\mathcal{T}_\theta$ is the map from model parameters $\epsilon, t, \nu$, and $l$ to LDOS images through \eqref{equ:ldos}-\eqref{eq:LDOSgrid} at twist angle $\theta$. Continuity follows from Theorem \ref{thm:invcon} and continuity of $\mathcal{T}_\theta$. \end{proof}

\begin{remark}
\label{rmk:twist}
It is natural to ask about stability of the twist operator $\tilde{\mathcal{L}}_\theta$. Since $\tilde{\mathcal{L}}_\theta$ is the composition of $\mathcal{T}_\theta$ and $\mathcal{K}^{-1}$, and the map from model parameters to LDOS images $\mathcal{T}_\theta$ through \eqref{equ:ldos}-\eqref{eq:LDOSgrid} is smooth, this question reduces to stability of the inverse map $\mathcal{K}^{-1}$ already discussed in Theorem \ref{thm:in_specific} and Remarks \ref{rmk:invesesta} and \ref{rem:stab_analytic}. It is interesting to consider whether the twist operator could be constructed in a different way, in particular, without reference to the inverse map from LDOS images to model parameters. If this were possible, it could be that a much better stability result could be proved.
\end{remark}




\section{Universal approximation with neural network} \label{sec:approximation_by_NN}

With existence and continuity of the twist operator established in Theorem \ref{thm:twist_op_con}, we now turn to approximability of this operator by a neural network.
This will follow easily from the well-known universal approximation theorem \cite{chen1995universal,barron1993universal}, which establishes conditions under which neural networks can be used to approximate functions. For the reader's convenience, we state a simple special case of this result for continuous functions and two-layer neural networks \cite{hornik1989multilayer,Pinkus_1999}. This will suffice to guarantee that the twist operator between SD-LDOS images can be approximated by such a neural network under the assumptions of Theorem \ref{thm:twist_op_con}. In passing we remark that versions of this result have been proved for deep neural networks~\cite{Zhou2018CNNApprox}, smooth functions~\cite{lu2021deep}, and networks with different structures developed for approximating nonlinear continuous functionals and nonlinear operators~\cite{chen1995universal,lu2021learning}. 
\begin{theorem}{\textbf{Universal Approximation Theorem}}
\label{thm:universa2}
    Let $C(X,\mathbb{R}^n)$ denote the space of continuous functions from a subset $X \subset \mathbb{R}^n$ to $\mathbb{R}^n$. Let $\sigma \in C(\mathbb{R}, \mathbb{R})$ be a non-polynomial function, and let $\sigma \circ x$ denote $\sigma$ applied to each component of $x$, so that $(\sigma \circ x)_i = \sigma(x_i)$. Then, for any $f\in C(X,\mathbb{R}^n)$ and $\epsilon > 0$, there exist $k \in \mathbb{N}_{>0}$, $A \in \mathbb{R}^{k \times n}$, $b \in \mathbb{R}^k$, and $C \in \mathbb{R}^{n \times k}$ such that
\begin{equation}
    \label{eq:NN_approx}
    \sup_{x \in X} \| f(x) - g(x) \| < \epsilon,
\end{equation}
where
\begin{equation}
    \label{eq:g_approx}
    g(x) = C \cdot (\sigma \circ (A \cdot x + b)).
\end{equation}
\end{theorem}
For the proof of Theorem \ref{thm:universa2}, see \cite{Pinkus_1999}. 
When an estimate of the form \eqref{eq:NN_approx} holds, we say that the function $f$ can be approximated by a two-layer neural network. Neural networks with additional layers are defined recursively from \eqref{eq:g_approx} by applying additional affine transformations (with associated coefficients $A, b$) composed with $\sigma$ to the output of the previous layer, so that the number of layers corresponds to the number of times $\sigma$ appears in \eqref{eq:g_approx} plus one \cite{doi:10.1137/18M1165748}. It now follows immediately from Theorem \ref{thm:twist_op_con} that two-layer neural networks can approximate the twist operator if the existence condition for either \cref{thm:in_specific} or \cref{thm:in_general} is satisfied. We summarize this in the following theorem.
\begin{theorem}{\textbf{Universal Approximation Theorem for Twist Operators (fixed $\theta$)}} \label{thm:twistuni_fix} Fix $0 \leq \theta < 1$. Assume the conditions either of Theorem \ref{thm:in_specific} or of Theorem \ref{thm:in_general}, and let $\mathcal{K}$ be defined by $\mathcal{K}_1$ or $\mathcal{K}_2$ depending on the case. Let $\tilde{\mathcal{L}}_\theta : \range(\mathcal{K}) \subset \mathbb{R}^{N_d \times N_E} \rightarrow \mathbb{R}^{N_d \times N_E}$ denote the twist operator defined in Theorem \ref{thm:twist_op_con}. Then, for any $\epsilon > 0$, there exist $k \in \mathbb{N}_{>0}$, $A \in \mathbb{R}^{k \times N_d N_E}$, $b \in \mathbb{R}^k$, and $C \in \mathbb{R}^{N_d N_E \times k}$ such that
\begin{equation}
    \label{eq:NN_approx_2}
    \sup_{x \in \range{\mathcal{K}}} \| \tilde{\mathcal{L}}_\theta(x) - g(x) \| < \epsilon,
\end{equation}
where
\begin{equation}
    \label{eq:g_approx_2}
    g(x) = C \cdot (\sigma \circ (A \cdot x + b)).
\end{equation}
\end{theorem}
Theorem \ref{thm:twistuni_fix} guarantees that it is possible to use a two-layer neural network to approximate the local density of states for a fixed lattice mismatch (twist angle) $\theta$. Training the neural network refers to choosing $C, A$, and $b$ appearing in \eqref{eq:g_approx_2} to minimize the cost function defined by
\begin{equation}
    \frac{1}{N} \sum_{i = 1}^N \| \tilde{\mathcal{L}}_\theta(x_i) - g(x_i) \|,
\end{equation}
where $(x_i,\tilde{\mathcal{L}}_\theta(x_i))$ for $i = 1,...,N$ denotes $N$ sample pairs consisting of untwisted and twisted SD-LDOS images computed with the same tight-binding parameters (see Section \ref{sec:commen} for details of how the twisted SD-LDOS images are computed).

In practice, the neural network trained in \cite{Liu_2022} had more layers, and was trained using a multi-layer convolutional neural network (CNN) with the alternating direction method of multipliers (ADMM). Moreover, the data used in \cite{Liu_2022} is slightly more general than what we discussed here. In \cite{Liu_2022}, there are ten types of materials and the data of both input and output data was allowed to correspond to different energy windows (i.e., different choices of $E_1$ and $E_{N_E}$ in \eqref{eq:E_vals}) for each different material. We ignore these subtleties in the present work in order to clearly present the most important ideas. However, we expect that the existence of a more general twist operator and its approximation theorem can be established. We  leave this to future work. 

An interesting generalization of the method described in the present paper would be to learning the twist operator over a range of lattice mismatches $\theta$. This could be done by training the neural network on tensor data consisting of multiple LDOS images corresponding to multiple lattice mismatches. We hope to address this problem in future work.

It remains to discuss how to accurately numerically compute LDOS images for use as training data. We discuss this in the following section.

\section{Numerical details of computing LDOS images} \label{sec:commen}



In this section, we discuss the methods used in \cite{Liu_2022} to numerically compute the LDOS defined by \eqref{equ:ldos} and prove error estimates guaranteeing convergence of the methods. The key idea is to restrict attention to rational lattice mismatches $\theta$ and use a Bloch decomposition. Note that error estimates have already been proved for irrational $\theta$'s; see \cite{massatt2017electronic,massatt2017incommensurate}. 


\subsection{The commensurate coupled chain}
\label{sec:commensurate}
We start by reviewing again the lattice structure and tight-binding model in the case of a rational lattice mismatch. The model we introduce here will be the same as the model introduced in Section \ref{sec:models} (after restricting to rational $\theta$ and relabeling sites), but the notation we introduce will be more convenient for the numerical convergence analysis. If lattice mismatch $\theta$ is a rational number, let $p,q$ be the integers satisfying $q(1-\theta)=p$ where $p,q \in \mathbb{N}$ and $\mathrm{gcd}(p,q)=1$. The moiré scale is  $\frac{p}{q-p}$. The coupled chain $\mathcal{R}_{1}\cup \mathcal{R}_{2}$ is periodic on a supercell lattice,
\[
\scl_1=\scl_2=\scl=\left \{ (1-\theta)qn : n \in \mathbb{Z}  \right\}=\left \{ pn : n \in \mathbb{Z}  \right \}.
\]Equivalently, atomic orbitals can also be labeled by 
\begin{align*}
    \Omega^s= (\scl_1\times \sci_1) \cup (\scl_2\times \sci_2)=\scl\times (\sci_1\cup \sci_2).
\end{align*}
The unit cell for the coupled chain can then be defined as     
\begin{equation*}
    \Gamma=\left \{ (1-\theta)q \beta: \beta \in [0,1)  \right \}=\left \{  p\beta: \beta \in [0,1) \right\},
\end{equation*}
where the reciprocal lattice unit cell is defined as
\begin{align*}
    \Gamma^{*}= \left \{   \frac{2\pi}{p}\beta :\beta \in [0,1)     \right \}.
\end{align*}
The length of each supercell is $p$ which is identical to the moiré scale when $q=p+1$. We can equivalently describe the multilattices $\mathcal{R}_{1}$ and $\mathcal{R}_{2}$
by lattices with shifts
\begin{gather*}
    \scl_1=\scl_2=\left \{ (1-\theta)qn : n \in \mathbb{Z}  \right \},\\  
    \sci_1:=\{0, 1,  \dots,p-1\}\times \mathcal{A}_1,\\
    \sci_2:=\{0, 1, \dots,q-1\} \times \mathcal{A}_2,
    \\ \stalo=n_1 + \ctalo \ \text{ for } \alpha^{s}_1=(n_1,\alpha_1)\in \sci_1,\\\stalt =n_2 + \ctalt \ \text{ for } \alpha^{s}_2=(n_2,\alpha_2)\in \sci_2.
\end{gather*}
$\sci_{i}$ denotes the set of indices of atomic orbitals within each supercell of chain $i$. The location of each orbital within each supercell is described by the shift $\tau_{\alpha}.$ We note that more complex 1-D coupled chain geometry for the basis atoms in each unit cell can also be described similarly. Detailed construction of the complex coupled chain and more examples are given in Appendix A of \cite{Liu_2022}.

    \begin{figure}[ht]
    \centering
    \includegraphics[width=0.8\linewidth]{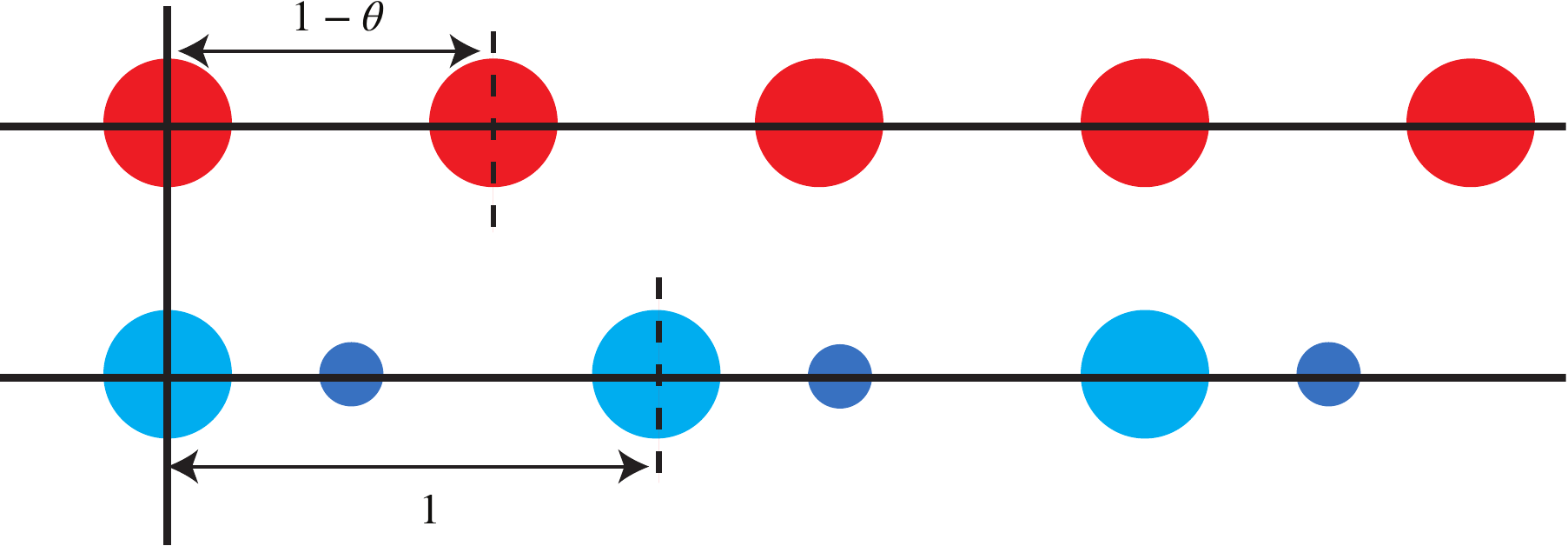}
        \caption{Schematic of the commensurate coupled chain system described in Section \ref{sec:commensurate} with two orbitals per cell ($m_1 = 1, m_2 = 2$), with a rational lattice mismatch $\theta = \frac{1}{3}$ without interlayer shift ($d = 0$). Despite the lattice mismatch, the system is periodic with supercell period $2$.}
    \label{fig:my_label}
\end{figure}

\subsection{Tight binding model for commensurate coupled chain}
When $\theta$ is a rational number, the coupled chain model is periodic and is referred to as a commensurate structure. It can be equivalently described by another set of indices $$\Omega^{s}=\scl_1\times \sci_1 \cup \scl_2\times \sci_2=\scl \times \sci $$where $\sci=\sci_1 \cup \sci_2$. The space for wave function of electrons can then be defined as
\begin{equation}\label{eq:hop}
     \ell^2(\Omega^s) = \left\{ (\psi(R_i\sigma))_{R_i\sigma \in \Omega^s } : \sum_{R_i\sigma \in \Omega^s}| \psi(R_i\sigma)|^2 < \infty \right\}.
    \end{equation}
The intralayer hopping energy for the coupled chain where $R \alpha, R'\alpha'\in \mathcal{R}^s_i\times \mathcal{A}^s_i$ can be defined as
\begin{equation}
    \label{equ:commensurateH}
    H(R\alpha,R'\alpha') = \begin{cases}
       T_{\alpha,\alpha'},\, & R-R'=p, \\
       T_{\alpha',\alpha},\, & R-R'=-p, \\
       \mathcal{E}_{\alpha,\alpha'}, \, &R=R',
       \\ 0, \, &\text{else},
    \end{cases}
\end{equation}
where $T_{\alpha,\alpha'}$ and $\mathcal{E}_{\alpha,\alpha'}$ are defined through $t_{\alpha,\alpha'}$ and $\epsilon_{\alpha,\alpha'}$ appearing in \eqref{eq:hop0} in such a way that the supercell Hamiltonian equals the tight-binding Hamiltonian introduced in Section \ref{sec:models} up to relabeling sites. 

The interlayer hopping energy is modeled by the interlayer hopping function $h$
\begin{equation*}
    \label{equ:Cinterlayer}
    H(R\alpha,R'\alpha')=h^s_{\alpha,\alpha'}(|r^s|),  \quad r=R+\tau_{\alpha}+d(-\delta_{i1}+\delta_{i2})-R'-\tau_{\alpha'}
\end{equation*}
where $R \alpha \in \mathcal{R}^s_i\times \mathcal{A}^s_i$,  $R'\alpha'\in \mathcal{R}^s_j\times \mathcal{A}^s_j$ and $i\neq j$. We again define each component of $h^{s}$ so that the supercell Hamiltonian equals the tight-binding Hamiltonian introduced in Section \ref{sec:models}. In particular, each component is a smooth function that decays exponentially.

\subsection{Bloch transform for commensurate coupled chain}
The model defined in this section is periodic. Many properties of such models, such as the stacking-dependent local density of states, are most efficiently calculated via the Bloch transform, defined as follows. The Bloch transform $\mathcal{G}:\ell^{2}(\scl\times \sci) \to C(\Gamma^{*}) \times \sci$ can be defined by
\begin{align} \label{eq:G}
    [\mathcal{G}\psi]_{\alpha}(k)=\widecheck{\psi}_{k}(\alpha):=\frac{1}{|\Gamma^{*}|^{\frac{1}{2}}}\sum_{R\in \scl} \psi(R\alpha) e^{-i k \cdot (R+\tau_{\alpha})}, \ \alpha \in \sci,
\end{align}
and its inverse by 
\begin{align}
    \label{equ:Ginverse}[\mathcal{G}^{-1}\widecheck{\psi}](R,\alpha)=\frac{1}{|\Gamma^{*}|^{\frac{1}{2}}} \int_{\Gamma^{*}} e^{ik \cdot (R+\tau_{\alpha})}\widecheck{\psi}_{k}(\alpha) \, \text{d}k, \ R\in \scl \  \text{and} \  \alpha \in \sci.
\end{align}
We can apply the Bloch transform to $H\psi$ for $\psi \in \ell^2(\scl \times \sci )$ to derive 
\begin{equation}
\label{equ:BtransH}
    [\mathcal{G}H\psi](k)=H(k)[\mathcal{G}\psi] (k)
\end{equation}
where $H(k)=[\mathcal{G}H\mathcal{G}^{-1}](k)$ is a $|\sci|\times|\sci|$ Hermitian matrix. Details of the derivation of $H(k)$, known as the Bloch Hamiltonian, can be found in \cref{sec:Hk} and \cite{massatt2017incommensurate}. 

\subsection{Stacking-dependent local density of states for the commensurate coupled chain}

The stacking-dependent local density of states can always be defined by \eqref{equ:ldos}. When the lattice mismatch $\theta$ is rational, this quantity can be efficiently numerically computed using the Bloch transform. Let $H_{\theta,d}$ denote the commensurate coupled chain Hamiltonian introduced in the previous section, where we again use subscripts to emphasize the dependence of the Hamiltonian on $\theta$ and $d$. Since $g$ is a polynomial, we can express the stacking-dependent LDOS \eqref{equ:ldos} of $H_{\theta,d}$ using the Bloch transform \eqref{eq:G} as
\begin{equation} \label{eq:LDOS_momspace}
    \rho_{\theta}(d,E) = \frac{1}{|\Gamma^*|} \sum_{\alpha : R\alpha \in \mathcal{I}_{2,d}} \int_{0}^{\frac{2 \pi}{p}} g( H_{\theta,d}(k)-E)_{\alpha,\alpha} \ \text{d}k,
\end{equation}
where we define $H_{\theta,d}(k) := [ \mathcal{G} H_{\theta,d} \mathcal{G}^{-1} ](k)$. With a slight abuse of notation, we denote the supercell orbital indices $\alpha \in \mathcal{A}^s$ which correspond to the indices $R\alpha \in \mathcal{I}_{2,d}$ in \eqref{equ:ldos} by $\alpha : R \alpha \in \mathcal{I}_{2,d}$. We provide the derivation of \eqref{eq:LDOS_momspace} from \eqref{equ:ldos} in \cref{sec:repre_sec}. It is then natural to approximate \eqref{eq:LDOS_momspace} by discretizing the integral over the Brillouin zone and truncating the interlayer hopping in real space (for more details, see \cref{sec:krepresent2}. We use $H_{k,\text{trunc}}$ to denote the Bloch transformed truncated Hamiltonian at wave-number $k$ (we suppress the subscripts $d, \theta$ since they do not play a role here). We thus arrive at the following numerical approximation of the stacking-dependent LDOS 
\begin{equation} \label{eq:LDOS_momspace_trunc}
    \varrho_\theta(d,E) := \sum_{\alpha: R\alpha \in \mathcal{I}_{2,d}} \sum_{i=1}^{M} g_n(H_{k_i,\text{trunc}}-E)_{\alpha,\alpha} \Delta k,
\end{equation}
where $\Delta k =\frac{2\pi}{p M}$, $k_i=\frac{2\pi i}{p M}$ for $i=1, 2,\dots, M$, and $M$ and $p$ are positive integers. The numerical scheme is equivalent to \eqref{equ:ldos} for a interlayer coupling truncated system with $M$ supercells.

\begin{theorem}
\label{thm:krepresent}
    Let 
    $\rho_\theta$ and $\varrho_\theta$ be as in \eqref{eq:LDOS_momspace} and \eqref{eq:LDOS_momspace_trunc}, respectively.
%
    Then, we have
    \begin{equation} \label{eq:err_est}
    |\rho_{\theta}(d,E) - \varrho_{\theta}(d,E)|\leq \left( \frac{C_1}{pe^{p\gamma M/2}-p} +\frac{C_{2}e^{-\gamma cp}}{\gamma p}\right) \cdot \frac{|\mathcal{I}_{2,d}|}{|\Gamma^{*}|} ,
\end{equation}
    where $\gamma$ is defined same as in \eqref{equ:h_cond}, $M$ is the length of the supercell and  $c>3$ is the number of supercells beyond which interlayer coupling will be truncated.
    The constant $C_1$ depends on $g, h, E, d$ and the constant $C_2$ depends on $h, N, E,d$.
\end{theorem} 
\begin{proof}
See proof at \cref{sec:krepresent2}.
\end{proof}

The first term on the right-hand side of \eqref{eq:err_est} corresponds to numerical quadrature errors and the second term corresponds to the interlayer coupling approximation error between $H_{\theta,d}$ and $H_{k,p}$.

We discuss some applications of Theorem \ref{thm:krepresent}. First, it shows that formula \eqref{eq:LDOS_momspace_trunc} can be used to generate accurate input data for training the neural network to approximate the twist operator. Second, combining Theorem \ref{thm:krepresent} with continuity of the twist operator established in Theorem \ref{thm:twist_op_con}, we see that the difference between the outputs of the twist operator applied to numerically-computed (through \eqref{eq:LDOS_momspace_trunc}) LDOS images and to exact LDOS images must be small. Finally, combining this result with the fact that the twist operator can be accurately approximated by a neural network in Theorem \ref{thm:twistuni_fix}, we have that this property is also shared by the neural network approximation to the twist operator. 

To be more precise, if we use $\mathcal{L}_{NN}$ to denote the neural network approximating twist operator, then the error between neural network prediction and exact LDOS image is 
\begin{equation}
    \begin{split}
        |\rho_{\theta}-\mathcal{L}_{NN}(\varrho_{0})|&= |\rho_{\theta}-\mathcal{L}_{NN}(\rho_{0})+\mathcal{L}_{NN}(\rho_{0})-\mathcal{L}_{NN}(\varrho_{0})|\\
        &\leq |\rho_{\theta}-\mathcal{L}_{NN}(\rho_{0})| + |\mathcal{L}_{NN}(\rho_{0})-\mathcal{L}_{NN}(\varrho_{0})|\\
        &\leq I_1 + C_{NN} \cdot I_2,
    \end{split}
\end{equation}
where $C_{NN}$ is a constant depending on the final trained neural network and $I_1, I_2$ are errors due to neural network approximation in \cref{thm:twistuni_fix} and numerical error in data generation in \cref{thm:krepresent}.

We summarize the discussion over numerical analysis of the end-to-end operator learning based scheme for the LDOS in Theorem~\ref{thm:endtoendanalysis}. 

\begin{theorem}{\textbf{Error Analysis for Learning the Twist Operator.}}
 Fix $0 \leq \theta < 1$. Assume the conditions either of Theorem \ref{thm:in_specific} or of Theorem \ref{thm:in_general}, and let $\mathcal{K}$ be defined by $\mathcal{K}_1$ or $\mathcal{K}_2$ depending on the case. 

 Then, for any $\epsilon > 0$, there exist $k \in \mathbb{N}_{>0}$, $A \in \mathbb{R}^{k \times N_d N_E}$, $b \in \mathbb{R}^k$, and $C \in \mathbb{R}^{N_d N_E \times k}$ such that
\begin{equation}
    \label{eq:NN_approx_3}
     \| \rho_{\theta} - \mathcal{L}_{NN}(\varrho_0) \| < \epsilon + \left( \frac{C_{NN}C_1}{pe^{p\gamma M/2}-p} +\frac{C_{NN}C_{2}e^{-\gamma cp}}{\gamma p}\right) \cdot \frac{|\mathcal{I}_{2,d}|}{|\Gamma^{*}|} ,
\end{equation}
for any
    $\rho_\theta$ be as in \eqref{eq:LDOS_momspace} and  $\varrho_0 \in \range(\mathcal{K}) $ be as in \eqref{eq:LDOS_momspace_trunc} for $\theta=0$, where $\mathcal{L}_{NN}$ denotes the two-layer neural network
\begin{equation}
    \label{eq:g_approx_3}
    \mathcal{L}_{NN}(x) = C \cdot (\sigma \circ (A \cdot x + b)).
\end{equation}
The $C_{NN}$ is a constant depending on the final trained neural network. Other parameters $p,\gamma,M,|\mathcal{I}_{2,d}|,|\Gamma^{*}|, C_1, C_2$ are defined similarly as in Theorem~\ref{thm:krepresent}.

\label{thm:endtoendanalysis}
\end{theorem}
\begin{proof}
    The proof follows directly from Theorem~\ref{thm:twist_op_con}, Theorem~\ref{thm:twistuni_fix} and Theorem~\ref{thm:krepresent}.
\end{proof}

\begin{remark}
    While we have discussed the stability of the inverse map $\mathcal{K}_1^{-1}$ and the twist operator in Remark \ref{rmk:invesesta} and Remark \ref{rmk:twist}, it remains unclear whether the trained neural network will exhibit stability or how its stability relates to that of the inverse map and the twist operator. This is not a trivial question: it is known that trained neural networks may be unstable in the worst case for certain tasks \cite{colbrook2022difficulty}. Since this issue occurs much more generally than in the specific problem we consider here, we do not investigate this question further and refer the interested reader to \cite{colbrook2022difficulty,gottschling2025troublesome, tang2024chebnet,zhang2014comprehensive} and the references therein.
\end{remark}

\section{Conclusion}  \label{sec:conc}

In this work, we have studied the twist operator, a map which relates properties of a twisted bilayer material to those of the same bilayer material without twist, focusing in particular on the specific case of the stacking-dependent local density of states (SD-LDOS). We have discussed the relevance of an inverse problem relating SD-LDOS data to the coefficients of the tight-binding model which generated it. We proved that under two types of assumptions, the inverse problem is well-posed and existence of the twist operator can be proved. We then invoked the Universal Approximation Theorem to prove that the twist operator can be approximated by a two-layer neural network and proved smallness of numerical errors in computing SD-LDOS images.


We hope that the present work will stimulate further research into predicting properties of twisted bilayers from data about their untwisted counterparts. We expect this approach to be especially useful when direct numerical computation of properties of the twisted bilayers is intractable; for example, for properties depending delicately on the effects of electron-electron interactions. We expect that this approach could even be used directly on experimental data collected from untwisted bilayers to predict the results of experiments on their twisted counterparts. Although our rigorous results do not apply in these cases, we hope that the fact that similar ideas can be made rigorous in simpler cases provides further impetus to such studies.


We close by noting some interesting potential future directions for further mathematical research. First, it will be interesting to investigate learning the parameters of tight-binding models using spectral data more generally, not only in the context of moir\'e materials. Then, it would be interesting to further develop the rigorous theory of this learning problem, i.e., to investigate well-posedness of the associated inverse problem. It may be possible to connect to the extensive literature on the inverse spectral problem for 1D Schr\"odinger operators; see, e.g., \cite{8a839ae8-b634-3548-a204-4d3ee1bc7c36}. It would be especially exciting to try to extend the ideas of \cite{Liu_2022} and the present work to 2D materials, although we expect proving well-posedness of the analogous inverse problem to be very difficult in this case. It would also be interesting to develop the rigorous theory of aspects of \cite{Liu_2022} not considered here, such as the generalization to multiple types of materials.

\bibliographystyle{siamplain}

\bibliography{references}

\appendix

\section{Proof for Existence of Twist Operator}

\subsection{Proof of Theorem \ref{thm:in_specific}}
\label{sec:in_specific}
\begin{proof}
Our strategy will be to first prove that LDOS images uniquely specify moments of the Hamiltonian evaluated in the home cell $\Gamma_{2,d}$ \eqref{eq:home_cell}. We will then prove that these moments uniquely specify model parameters. 

The polynomial $g$ used to compute LDOS $\rho_{\theta}(d,E)$ \eqref{equ:ldos} can be represented as
\begin{equation}
    g(x)=a_{\n} x^{\n}+a_{\n-1}x^{\n-1}+\dots +a_{0}.
\end{equation}
Then $g(x-E)$ can be expanded as 
\begin{align*}
    g(x-E)&=a_{\n} (x-E)^{\n}+a_{\n-1}(x-E)^{\n-1}+\dots +a_{0}
    \\&=a_{n_l} (x-E)^{n_l}+a_{n_{l-1}}(x-E)^{n_{l-1}}+\dots +a_{n_{0}}(x-E)^{n_{0}}\\
    &=\sum_{\mu=0}^{\n} f_{\mu}(E) x^\mu,
\end{align*}
where $\n=n_l>n_{l-1}\dots >n_{0}\geq 0$ and for any other $n'\neq n_{0}\dots n_{l}$, $a_{n'}=0$. 
Collecting the coefficients for the term $x^{k}$, we derive
\begin{align}
    f_{\mu}(E)=\sum_{\mu\leq n_m \leq \n} a_{n_m}C_{n_m}^{\mu}(-E)^{n_m-\mu}.
\end{align}
The highest degree monomial with respect to $E$ of $f_{\mu}(E)$ is $E^{\n-\mu}$. Since the coefficient of the highest degree term is nonzero, the $f_{\mu}$ is an $\n-\mu$ degree polynomial with respect to $E$.
\begin{lemma}
\label{lemma:inverseg}
Assume $f_{\mu}(E)$ is an $\n-\mu$ degree polynomial with respect to $E$ and the matrix $S$ is defined by
$$    (S)_{ij}=f_{j}(E_i), $$
$i=0, 1\dots \n$ and $j=0, 1\dots \n$ where $\n$ is the degree of polynomial, $E_i\in \mathbb{R}$ and $E_i\neq E_j$ for $i\neq j$. Then, matrix $S$ is a non-singular matrix. 
\end{lemma}

\begin{proof}
See Proof at \cref{sec:inversegproof}
\end{proof}

Therefore, local density of states for the specific 1-D system $\rho_{\theta}(d,E)$ can be represented as 
\begin{equation}
     \sum_{\mu=0}^{\n} f_{\mu}(E) \osum H_{\theta,d}^{\mu}(R\alpha,R\alpha),
\end{equation}
where $f_{\mu}(x)$ is a polynomial of degree $\n-\mu$ depending on polynomial $g$ in \eqref{equ:ldos}. We view $\osum H_{\theta,d}^{\mu}(R\alpha,R\alpha)$ as unknowns and $f_{\mu}(E)$ as coefficients for $\n+1$ different $E$ values. 
Given a local density of states image, \cref{lemma:inverseg} shows when $N_E > \n $, the linear system below has a unique solution
\begin{equation}
\begin{split}
     \sum_{\mu=0}^{n_{poly}} f_{\mu}(E_0) \osum H_{\theta,d}^{\mu}(R\alpha,R\alpha)&=\rho_{\theta}(d,E_{0}),\\
     \sum_{k=0}^{n_{poly}} f_{\mu}(E_1) \osum H_{\theta,d}^{\mu}(R\alpha,R\alpha)&=\rho_{\theta}(d,E_{1}),\\
     \dots \quad \dots\\
     \sum_{\mu=0}^{n_{poly}} f_{\mu}(E_{N_E}) \osum H_{\theta,d}^{\mu}(R\alpha,R\alpha)&=\rho_{\theta}(d,E_{N_{E}}).
\end{split}
\end{equation}
That is to say, given an LDOS image we can solve for moments of the Hamiltonian evaluated in the home cell, i.e.,
\begin{equation}
    \osum H_{\theta,d}^{\mu}(R\alpha,R\alpha),
\end{equation}
for any integer $0 \leq \mu \leq \n$.

It is natural to ask whether the linear system we have to solve, to recover the moments of the Hamiltonian, is stable with respect to changes in the LDOS image. This can be directly measured by estimating $|S^{-1}|_{\infty}$. We now provide a bound on this norm by re-writing $S$ in terms of a Vandermonde matrix and then using a known formula for the norm of the inverse of a Vandermonde matrix.

First of all, it can be noticed that for any $i$, $S_{i, n_{poly}}=f_{n_{poly}}(E_i)=a_{n_{poly}}$. The constant terms in other columns of $S$ matrices are $a_{k}$. Therefore, we can eliminate the constant terms in the $k$-th column by subtracting $\frac{a_k}{a_{n_{poly}}}$ times 
the $n_{poly}$-th column within the $S$ matrix. In another words, we can find a matrix $R_{n_{poly}}$ such that $S R_{n_{poly}}$ does not have constant term except the last column. The $S$ matrix can be represented as
\begin{equation}
    S=S \cdot R_{n_{poly}} \cdot R_{n_{poly}}^{-1}
\end{equation}
where $\cdot$ is used to denote the matrix-matrix multiplication. The matrix $R_{n_{poly}}$ equals

\begin{equation}
    \begin{pmatrix}

1 & 0 & \cdots & 0 \\

0 & 1 & \cdots & 0 \\

\vdots & \vdots & \ddots & \vdots \\

-\frac{a_1}{a_{n_{poly}}}& -\frac{a_2}{a_{n_{poly}}} & \cdots & 1

\end{pmatrix}
\end{equation}
and its inverse $R_{n_{poly}}^{-1}$ equals

\begin{equation}
    \begin{pmatrix}

1 & 0 & \cdots & 0 \\

0 & 1 & \cdots & 0 \\

\vdots & \vdots & \ddots & \vdots \\

\frac{a_1}{a_{n_{poly}}}& \frac{a_2}{a_{n_{poly}}} & \cdots & 1

\end{pmatrix}.
\end{equation}
Note that since
\begin{equation}
    |R_{n_{poly}}|_{\infty} \leq \sum_{i=0}^{n_{poly}} \left |\frac{a_i}{a_{n_{poly}}} \right|,
\end{equation}
 we can derive that
\begin{equation}
    |S^{-1}|_{\infty}= |R_{n_{poly}}(SR_{n_{poly}})^{-1}|_{\infty}\leq |R_{n_{poly}}|_{\infty} |(SR_{n_{poly}})^{-1}|_{\infty}.
\end{equation}
With same process, we could use a sequence of matrices to transform $S$ into  Vandermonde matrices $\tilde{S}:=S R_{n_{poly}} R_{n_{poly}-1} \dots R_{0}$, i.e.,

\begin{equation}
\begin{split}
    S &= S R_{n_{poly}} R_{n_{poly}-1} \dots R_{0} R_{0}^{-1} \dots R_{n_{poly}}^{-1} \\
    &= \tilde{S} R_{0}^{-1} \dots R_{n_{poly}}^{-1}.
\end{split}
\end{equation}

We can then estimate $|S^{-1}|_{\infty}$ by
\begin{equation}
\begin{split}
    |S^{-1}|_{\infty} &= |R_{n_{poly}} R_{n_{poly}-1} \dots R_{0} \tilde{S}^{-1}|_{\infty} \\ &\leq |R_{n_{poly}}|_{\infty} |R_{n_{poly}-1}|_{\infty} \dots |R_{0}|_{\infty} |\tilde{S}^{-1}|_{\infty} .
    \end{split}
\end{equation}
We note that for the matrix $R_j$ where $0\leq j \leq n_{poly}$, 

\begin{equation}
    |R_j|_{\infty} = \frac{\sum_{i=0}^{n-j} |a_{n-i} C_{n-i}^{n-j-i}|}{|a_n C_{n}^{n-j}|}
\end{equation}
by direct calculation. We observe that the $\tilde{S}$ equals 

\begin{equation}
    \begin{pmatrix}

a_{n_{poly}} C_{n_{poly}}^{0} (-E_0)^{n_{poly}} & a_{n_{poly}} C_{n_{poly}}^{1} (-E_0)^{n_{poly}-1} & \cdots &  a_{n_{poly}} \\

a_{n_{poly}} C_{n_{poly}}^{0} (-E_1)^{n_{poly}} & a_{n_{poly}} C_{n_{poly}}^{1} (-E_1)^{n_{poly}-1} & \cdots &  a_{n_{poly}} \\

\vdots & \vdots & \ddots & \vdots \\

a_{n_{poly}} C_{n_{poly}}^{0} (-E_{n_{poly}})^{n_{poly}}& a_{n_{poly}} C_{n_{poly}}^{1} (-E_{n_{poly}})^{n_{poly}-1} & \cdots & a_{n_{poly}}
\end{pmatrix}.
\end{equation}

Recall the estimate for the norm of inverse Vandermonde matrix~\cite{gautschi1974norm}, we conclude that

\begin{equation} \label{eq:S_inv_norm}
    |S^{-1}|_{\infty} \leq \frac{1}{2}\prod_{j=0}^{n_{poly}} \left( \sum_{i=0}^{n-j} |a_{n-i} C_{n-i}^{n-j-i}| \right) \max_{i}\left \{ \left(1+\frac{1}{E_i} \right)\prod_{i\neq j} \frac{1+E_j^2}{|E_i^2-E_j^2|}\right \},
\end{equation}
since $E_i + E_{N_E-i}=0$ for all $0\leq i \leq N_E$ and $N_E$ is odd. Equation \ref{eq:S_inv_norm} shows that the problem becomes unstable when the energy grid is too fine or when $n_{poly}$ is too large, which happens for very sharply peaked approximations of the delta function (the function $g$) in the LDOS formula \eqref{equ:ldos}.


We will next show that when $\theta=0$, we can recover the tight binding parameters for the 1-D coupled chain model defined in \cref{sec:twist_theory} from these moments. 

The $\mu=0$ moment yields $\osum H_{\theta,d}^{0}(R\alpha,R\alpha)=|\mathcal{A}| = 1$. 
Evaluating these moments for $\mu=1,2$, we have 
\begin{equation}
\begin{split}
    \osum H^{1}(R\alpha,R\alpha)&= \epsilon \\
     \osum H^{2}(R\alpha,R\alpha)&= \epsilon^2+2t^2+\nu^2 \sum_{n=-\infty}^{\infty}e^{-\frac{2|d-n|}{l}}\\
      &= \epsilon^2+2t^2+\nu^2(e^{-\frac{2d}{l}}+\sum_{n=1}^{\infty}e^{-\frac{2d+2n}{l}}+\sum_{n=1}^{\infty}e^{-\frac{2n-2d}{l}})\\
      &=\epsilon^2+2t^2+\frac{\nu^2(e^{\frac{2d-2}{l}}+e^{\frac{-2d}{l}})}{1-e^{-\frac{2}{l}}},
\end{split}
\end{equation}
where we used $l > 0$ (by assumption) to sum the geometric series in the last line. We see that the first moment uniquely specifies $\epsilon$. We will now prove that the second moment, evaluated at $d = 0, \frac{1}{4}$, and $\frac{1}{2}$, uniquely specifies $t, \nu$, and $l$. We define 
\begin{equation}
\label{equ:statnu}
    S_{t,\nu,l}(d):= 2t^2+\frac{\nu^2(e^{\frac{2d-2}{l}}+e^{\frac{-2d}{l}})}{1-e^{-\frac{2}{l}}},
\end{equation} 
and it suffices to prove there exists a one to one mapping between $(t,\nu,l)$ and value of $S$ on $d_1=0, d_2=\frac{1}{4}, d_3=\frac{1}{2}$. From direct calculation of the derivative, we have
\begin{equation}
    \begin{split}
        0 < d < \frac{1}{2} \implies \partial_{d} S <0, \\
        \frac{1}{2} < d < 1 \implies \partial_{d} S >0. \\
    \end{split}
\end{equation}
Moreover, since $l > 0$, $\frac{S_{t,\nu,l}(d_2)-S_{t,\nu,l}(d_1)}{S_{t,\nu,l}(d_3)-S_{t,\nu,l}(d_2)}=e^{-\frac{1}{2l}}+e^{\frac{1}{2l}}+1$ is a well-defined and strictly monotonic decreasing function with respect to $l$, i.e
\begin{equation}
     \frac{S_{t,\nu,l}(d_2)-S_{t,\nu,l}(d_1)}{S_{t,\nu,l}(d_3)-S_{t,\nu,l}(d_2)}= \frac{e^{-\frac{3}{2l}}+e^{-\frac{1}{l}}-e^{-\frac{2}{l}}-1}{2e^{-\frac{1}{l}}-e^{-\frac{3}{2l}}-e^{-\frac{1}{l}}}.
\end{equation}
We assume that  $(t,\nu,l)$, $(t',\nu',l')$ are such that
\begin{equation}
\label{equ:S3d}
\begin{split}
    S_{t,\nu,l}(d_1)&=S_{t',\nu',l'}(d_1), \\
    S_{t,\nu,l}(d_2)&=S_{t',\nu',l'}(d_2), \\
    S_{t,\nu,l}(d_3)&=S_{t',\nu',l'}(d_3),
\end{split}
\end{equation}
we will prove that this implies that $t = t', \nu = \nu', l = l'$.
Equation \eqref{equ:S3d} implies that 
\begin{equation}
    \frac{S_{t,\nu,l}(d_2)-S_{t,\nu,l}(d_1)}{S_{t,\nu,l}(d_3)-S_{t,\nu,l}(d_2)}=\frac{S_{t',\nu',l'}(d_2)-S_{t',\nu',l'}(d_1)}{S_{t',\nu',l'}(d_3)-S_{t',\nu',l'}(d_2)}.
\end{equation}
But now we must have that $l=l'$ since the quotient function is a single variable monotonic function. In addition to that, $S_{t,\nu,l}(d_1)=S_{t',\nu',l'}(d_1)$ and $S_{t,\nu,l}(d_2)=S_{t',\nu',l'}(d_2)$ gives $t=t'$ and $\nu=\nu'$ since the parameters are assumed to be positive. 
To be specific, the two conditions and $l=l'$
give 
\begin{equation}
\label{equ:tvsystem}
    \begin{bmatrix}
        2 & \frac{1+e^{-\frac{2}{l}}}{1-e^{-\frac{2}{l}}}\\
        2 & \frac{e^{-\frac{3}{2l}}+e^{-\frac{1}{l}}}{1-e^{-\frac{2}{l}}}
    \end{bmatrix}
    \begin{bmatrix}
        t^2-t'^2 \\ v^2-v'^2
    \end{bmatrix}
    = \begin{bmatrix}
        0 \\0
    \end{bmatrix}.
\end{equation}
Solving the linear system gives that $t^2=t'^2$ and $\nu^2=\nu'^2$. Given $t,t',\nu,\nu'$ are positive, we have that $t=t'$, $\nu=\nu'$.

We have thus proved that the choices of $\epsilon,t,\ell,\nu$ are uniquely determined by the LDOS images 
and, consequently, the one-to-one mapping $\mathcal{K}_1$
in the theorem is proved to be well-posed.

We can also consider the stability of the mapping from moments of Hamiltonian to tight binding parameters,
\begin{equation}
    M_{H}:=\left( \osum H^{1}(R\alpha,R\alpha), \osum H^{2}(R\alpha,R\alpha) \right)_{d_1,d_2,d_3} \to (\epsilon,t,\nu,l),
\end{equation}
where $(\cdot,\cdot)_{d_1,d_2,d_3}$ denotes a six dimensional vector consists of function values of different $d$ values. Assume $M_H$ is perturbed at most $\delta<1$ in $l_{\infty}$ norm and use $S_{t',\nu',l'}(d)$ to denote the perturbed $S_{t,\nu,l}$ function value at $d$, we have that for $d=d_0,d_1,d_2$,

\begin{equation}
    |S_{t,\nu,l}(d)-S_{t',\nu',l'}(d)|\leq \delta.
\end{equation}
Since the $\mathcal{P}_{comp}$ is a compact subset, direct calculation implies that

\begin{equation}
    C_{\mathcal{P},1} |l-l'| \leq \left |\frac{S_{t,\nu,l}(d_2)-S_{t,\nu,l}(d_1)}{S_{t,\nu,l}(d_3)-S_{t,\nu,l}(d_2)}-\frac{S_{t',\nu',l'}(d_2)-S_{t',\nu',l'}(d_1)}{S_{t',\nu',l'}(d_3)-S_{t',\nu',l'}(d_2)} \right| \leq C_{\mathcal{P},0} \delta. 
\end{equation}
where nonzero constant $C_{\mathcal{P},0}$ depends on the range of parameters of $t,\nu,l,\epsilon$ and nonzero constant $C_{\mathcal{P},1}$ depends on range of parameters of $l$. Note that the first order moment of Hamiltonian guarantee that $|\epsilon-\epsilon'|\leq \delta$. The stability of $t$ and $\nu$ can be derived from the stability of solving linear system for \cref{equ:tvsystem} with different $d$ values and the fact that $|t|,|\nu|$ has nonzero lower bound.

Indeed, we compute the eigenvalue of linear system \cref{equ:tvsystem} with general $d_1<d_2$ in order to estimate the stability of inverse map with respect to LDOS images. Define function $F(d)$ as
\begin{equation}
    F(d) = \frac{e^{\frac{2d-2}{l}}+e^{\frac{-2d}{l}}}{1-e^{-\frac{2}{l}}}.
\end{equation}
The two eigenvalues of the linear system~\cref{equ:tvsystem} equals to
\begin{equation}
    \lambda_{\pm} = \frac{2+F(d_2) \pm \sqrt{(F(d_2)-2)^2+8F(d_1)}}{2}.
\end{equation}
Note that $\lambda_{+}$ is positive and has nonzero lower bound. It suffices to consider the upper bound of $|\lambda_{-}^{-1}|$ as a measure for the stability. Through direct calculation, $|\lambda_{-}^{-1}| \leq \left| \frac{C}{d_2-d_1} \right |$ where the constant $C$ only depends on the range of $l$. The perturbed $l$ introduce $\mathcal{O}(\delta)$ error in the coefficient matrix and the moments of Hamiltonian introduce $\mathcal{O}(\delta)$ error in the right-hand side vector of the linear system. Recall a well-known stability result in numerical linear algebra, when solving $Ax=b$ the relative error can be bounded by

\begin{equation}
    \frac{|\Delta x|}{|x|} \leq \|A^{-1} \| \|\Delta A \| + \|A\| \|A^{-1}\|\frac{|\Delta b|}{|b|}.
\end{equation}
In our case $\|A^{-1}\|$ scales as $\frac{C}{d_2-d_1}$ and $\|A\|$ is uniformly bounded. Note that $|x|,|b|$ is bounded from below. In summary, we can derive that
\begin{equation}
    |\Delta x| \leq \frac{C' \delta}{d_2-d_1}
\end{equation}
for some constant $C'=\mathcal{O}(\frac{C_{\mathcal{P},0}}{C_{\mathcal{P},1}})$ which may blow up when $d_1, d_2, d_3$ are close to each other. This calculation shows clearly that the recovery of model parameters from moments of the Hamiltonian can be unstable if $d_1, d_2, d_3$ are chosen too close together. 

Note that instability can also occur regardless of the choice of $d_1, d_2, d_3$ if $l$ is too large. This can be seen from \eqref{equ:tvsystem}, where it is clear that the matrix rows become nearly parallel in this limit. We avoid this problem by assuming that $l$ is drawn from a compact, hence bounded, set.

In conclusion, when $d_1 = 0, d_2 = \frac{1}{4}, d_3 = \frac{1}{2}$, there exists a constant $C_{\mathcal{P}} \neq 0$ such that
\begin{equation} \label{eq:stab_2}
    \max \{ |\epsilon-\epsilon'|,|t-t'|,|\nu-\nu'|,|l-l'|  \} \leq C_{\mathcal{P}}\delta.
\end{equation}
However, one must be aware that the constant $C_{\mathcal{P}}$ may blow up if $d_1, d_2, d_3$ are chosen too close to each other. This concludes the discussion of stability for the inverse map from density of states images to model parameters: putting estimate \eqref{eq:S_inv_norm}, which quantifies the stability of the recovery process from LDOS images to moments of the Hamiltonian, together with estimate \eqref{eq:stab_2}, which quantifies the stability of recovering model parameters from Hamiltonian moments, we obtain the stability estimate \eqref{eq:K_1_inv}.

\end{proof}

\subsubsection{Proof of Lemma \ref{lemma:inverseg}}
\label{sec:inversegproof}
\begin{proof}
We prove the lemma by contradiction. If the matrix $S$ is singular, there exists a non trivial linear combination of column vectors of $S$, which equals to zero vector, i.e, 
\begin{align*}
    \sum_{i=0}^{n} c_{i}f_{i}(E)=0,
\end{align*}
where $c_i$ are not all zero and for any $E=E_j$ and $j\in \{0,1,2,3\dots,\n \}$. The linear system can be viewed as a linear system with Vandermonde matrix $V$ as coefficients and linear combination of $c$ as unknowns. To be specific, $(V)_{ij}=E_{i}^{j}$ and the matrix is non-singular. For the coefficients of highest degree of $E$, it only involve $c_0$. Therefore via solving the linear system, $c_0$ must be zero. Through the same argument, we can derive all $c$ are zero and therefore we find the contradiction.
\end{proof}

\subsection{Proof of Theorem \ref{thm:in_general}}
\label{sec:in_general}
\begin{proof}
With same argument in the proof of \cref{thm:in_specific}, we can show that given a local density of states image we can solve for
\begin{equation}
    \osum H_{\theta,d}^{\mu}(R\alpha,R\alpha),
\end{equation}
for any integer $0\leq \mu \leq \n$. Evaluating these moments we have 
\begin{equation}
\label{equ:Hfunc}
    \begin{split}
        \osum H_{\theta,d}(R\alpha,R\alpha)&= \epsilon \\
        \osum H_{\theta,d}^{2}(R\alpha,R\alpha)&= \epsilon^2+2t^{2}+ \sum_{n=-r_0}^{r_0} \nu^2 f^2\left(\frac{d-n}{l}\right).
    \end{split}
\end{equation}
We define 
\begin{equation}
\label{equ:S}
    S_{t,\nu,l}(d):= 2t^2+\sum_{n=-r_0}^{r_0} \nu^2 f^2\left(\frac{d-n}{l}\right).
\end{equation} 
    The parameter $\epsilon$ can be directly solved from first equation in \eqref{equ:Hfunc}. It suffices to prove there exists an integer $N_0$ such that when $N_{d}>N_0$, 
\begin{align}
    \Tilde{\mathcal{K}_2} :{(t,\nu,l)} \to ( S_{t,\nu,l}(d_1), S_{t,\nu,l}(d_2),\dots, S_{t,\nu,l}(d_{N_d})),
\end{align}
    is a one-to-one mapping from $\mathcal{P}_{fin}$ to its domain. Since $\mathcal{P}_{fin}$ is finite and $S$ is smooth, there exists a constant $c_1 \geq 0$ such that   
    $|\partial_{d}^2 S_{t,\nu,l}(d)|\leq c_1 $ for all $t,\nu,l \in \mathcal{P}_{fin}$ and all $0 \leq d \leq 1$. We first prove the following lemma.
\begin{lemma}
\label{lem:decayf}
If for any $d\in [0,1]$, \begin{equation}
    \partial_{d}S_{t,\nu,l}(d)=\partial_{d}S_{t',\nu',l'}(d),
\end{equation}  and there exists $d\in [0,1]$ such that 
\begin{equation}
    S_{t,\nu,l}(d)=S_{t',\nu',l'}(d),
\end{equation}
then $(t,\nu,l)=(t',\nu',l')$.
\label{lem:asyp}
\end{lemma}
\begin{proof}
Assume the result is not true, then there exists $(t,\nu,l)\neq(t',\nu',l')$ such that
\begin{equation}
    \displaystyle \max_{0\leq d\leq 1}|\partial_{d}S_{t,\nu,l}(d)-\partial_{d}S_{t',\nu',l'}(d)|=0.
\end{equation}
Since the function $\partial_{d}S_{t,\nu,l}(d)-\partial_{d}S_{t',\nu',l'}(d)$ is entire with respect to $d$, it follows that
    \begin{equation}
    \label{equ:S_derivative}
        \partial_d S_{t,\nu,l}(d) = \partial_{d}S_{t',\nu',l'}(d) \text{ for all } 0 \leq d < \infty.
    \end{equation}
    We start by showing how the assumption on $f$ in \eqref{eq:f_cond} implies $t=t'$. First, note that \cref{equ:S_derivative} together with the fact that the two functions are equal for some $d\in [0,1]$ implies that
    \begin{equation}
        S_{t,\nu,l}(d) = S_{t',\nu',l'}(d) \text{ for all } 0\leq d<\infty.
    \end{equation}
    That is to say,
    \begin{equation}
        2t^2+\sum_{n=-r_0}^{r_0} \nu^2 f^2\left(\frac{d-n}{l}\right)=2t'^2+\sum_{n=-r_0}^{r_0} \nu'^2 f^2\left(\frac{d-n}{l'}\right).
    \end{equation}
The exponential decay assumption of interlayer hopping function shows that when $d$ goes to infinity, the equation becomes
\begin{equation}
    t^2=t'^2 \implies t=t',
\end{equation}
 as $t$ and $t'$ are positive. We now show how the assumption also implies $l = l'$. If $l< l'$, the assumption \eqref{eq:f_cond} gives that for each $n$, 
 \begin{equation}
     \lim_{d\to \infty} \frac{f^2\left (\frac{d-n}{l}\right)}{f^2\left(\frac{d-n}{l'} \right)}=0.
 \end{equation}
This leads to the contradiction,
\begin{equation}
     \lim_{d\to \infty} \frac{\sum_{n=-r_0}^{r_0} f^2\left (\frac{d-n}{l}\right)}{\sum_{n=-r_0}^{r_0} f^2\left(\frac{d-n}{l'} \right)}=0 \neq \frac{\nu'^2}{\nu^2}
\end{equation}
as $\nu$ and $\nu'$ are nonzero. For the case $l>l'$, similar contradiction can be found. Therefore we prove that $l=l'$. Since $f$ is positive function, we can further derive $\nu=\nu'$, which contradicts with the statement $\nu \neq \nu'$.
\end{proof}


We now return to the proof that $\tilde{K}_2$ is one-to-one. Suppose that $(t',\nu',l ')\neq (t,\nu,l)$. By \cref{lem:decayf}, the following inequality holds 
\begin{equation} \label{equ:c0}
    c_0 := \displaystyle \min_{(\epsilon, t',\nu',l')\neq(\epsilon, t,\nu,l)\in \mathcal{P}_{\text{fin}}} \max_{0\leq d\leq 1}|\partial_{d}S_{t,\nu,l}(d)-\partial_{d}S_{t',\nu',l'}(d)| > 0.
\end{equation}
We also introduce the following
\begin{equation}
\label{equ:c1}
    c_1 := \displaystyle \max_{(\epsilon, t,\nu,l),\in \mathcal{P}_{\text{fin}} }\max_{0\leq d \leq 1} |\partial_d^2 S_{t,\nu,l}(d)| \geq 0.
\end{equation}
We now assume for a contradiction that 
\begin{align} \label{eq:tilde_K}
    \Tilde{\mathcal{K}_2}(t,\nu,l)=\Tilde{\mathcal{K}_2}(t',\nu',l ')
\end{align}
holds with $N_d > \max \{ 1, \frac{3 c_1}{c_0} \} $. First, note that \eqref{eq:tilde_K} gives 
\begin{align} \label{eq:thaat}
    \frac{S_{t,\nu,l}(d_{i+1})-S_{t,\nu,l}(d_{i})}{d_{i+1}-d_{i}}=\frac{S_{t',\nu',l'}(d_{i+1})-S_{t',\nu',l'}(d_{i})}{d_{i+1}-d_{i}}.
\end{align}
Let $\Delta d := \frac{1}{N_d}$ denote the spacing of $d$ grid (recall \eqref{eq:d_vals}). Then, by Taylor expansion, we have, for any adjacent grid points $d_i, d_{i+1}$,
\begin{align*}
    \frac{S_{t,\nu,l}(d_{i+1})-S_{t,\nu,l}(d_{i})}{\Delta d}-\partial_{d} S_{t,\nu,l}(d_i)\leq \frac{c_1}{2} \Delta d.
\end{align*}
Using uniform boundedness of $\partial_d^2 S$ again, we can deduce 
\[\left|\frac{S_{t,\nu,l}(d_{i+1})-S_{t,\nu,l}(d_{i})}{\Delta d}-\partial_{d} S_{t,\nu,l}(d)\right|\leq \frac{3c_1}{2} \Delta d\] for all $d\in [d_i,d_{i+1}].$
With the same argument, the same inequality can be established for $S_{t',\nu',l'}$. Using \eqref{eq:thaat} and the triangle inequality, these two inequalities give that for any $d\in [d_i,d_{i+1}]$,
\begin{align*}
    |\partial_{d} S_{t,\nu,l}(d)-\partial_{d} S_{t',\nu',l'}(d)|\leq \left |\frac{S_{t,\nu,l}(d_{i+1})-S_{t,\nu,l}(d_{i})}{\Delta d}-\partial_{d} S_{t,\nu,l}(d) \right |\\+ \left |\frac{S_{t',\nu',l'}(d_{i+1})-S_{t',\nu',l'}(d_{i})}{\Delta d}-\partial_{d} S_{t',\nu',l'}(d) \right|
    \leq 3c_1 \Delta d .
\end{align*}
By repeating the argument for $d$ in each interval $[d_i,d_{i+1}]$, the inequality holds for any $d\in [0,1]$. But now we have a contradiction of \eqref{equ:c0}, because, by our choice of $N_d$, we have, for any $d \in [0,1]$,
\begin{equation}
    c_0 \leq |\partial_{d} S_{t,\nu,l}(d)-\partial_{d} S_{t',\nu',l'}(d)| \leq 3c_1 \Delta d= 3c_1 \frac{1}{N_d} < c_0.
\end{equation}
\end{proof}

\subsection{Proof of Theorem \ref{thm:invcon}}
\label{sec:invcon}

\begin{proof}

The continuity for operator $\mathcal{K}_{2}^{-1}$ follows directly from the definition as this is a mapping from a discrete set to discrete set.

To prove the operator $\mathcal{K}_{1}^{-1}$ is continuous, first recall that the mapping $\mathcal{K}_{1}$ defined in Theorem \ref{thm:in_specific} maps a compact subset $\mathcal{P}_{\text{comp}}$ of $\mathbb{R}_{+}^{4}$ to the LDOS evaluated on a grid, 
\begin{equation}
    \mathcal{K}_1: 
        \mathcal{P}_{\text{comp}} \rightarrow \text{Range}(\mathcal{K}_1) \subset \mathbb{R}^{N_d\times N_E}.
\end{equation}
Since $\mathcal{K}_1$ is a bijection between $\mathcal{P}_{\text{comp}}$ and its range by \cref{thm:in_specific}, $\mathcal{K}_{1}^{-1}$ is well defined. We will show continuity by establishing that for any $M$ closed in $\mathcal{P}_{\text{comp}}$, we have that $(\mathcal{K}_1^{-1})^{-1}(M)$ is closed in $\text{Range}(\mathcal{K}_1)$.

Let $M \subset \mathcal{P}_{\text{comp}}$ be a closed set in the subspace topology of $\mathcal{P}_{\text{comp}}$. It follows that the set $M$ can be represented as $\Tilde{M} \cap  \mathcal{P}_{\text{comp}}$ where $\Tilde{M}$ is a closed and bounded set in $\mathbb{R}^{4}$. It follows that $\mathcal{K}_1(\Tilde{M})$ is also closed and bounded since $\mathcal{K}_1$ is continuous. We then have that $\mathcal{K}_1(M)=\mathcal{K}_1(\Tilde{M}\cap \mathcal{P}_{\text{comp}})=\mathcal{K}_1(\Tilde{M})\cap \mathcal{K}_1(\mathcal{P}_{\text{comp}})$, so that $\mathcal{K}_1(M)$ is closed in the subspace topology of $\mathcal{K}_1(\mathcal{P}_{\text{comp}}) = \text{Range}(\mathcal{K}_1)$.

\end{proof}

\section{Derivation of momentum space LDOS formula \cref{eq:LDOS_momspace} from \cref{equ:ldos}}
\label{sec:repre_sec}

We start by stating the derivation as a theorem.
\begin{theorem}
Suppose the Bloch transform $\mathcal{G}$ and its inverse $\mathcal{G}^{-1}$ is defined as in \cref{eq:G} and \cref{equ:Ginverse} respectively. Given a tight binding model Hamiltonian $H$ of a periodic system with period $p$, defined in \cref{equ:commensurateH}, the local density of states of the system defined in \cref{equ:ldos}
\begin{equation}
    \rho_{\theta}(d,E):=\osum g( H-E)(R\alpha,R\alpha),
\end{equation}
has a integral representation
\begin{equation}
    \rho_{\theta}(d,E) = \frac{1}{|\Gamma^*|} \sum_{\alpha: R\alpha \in \mathcal{I}_{2,d}} \int_{0}^{\frac{2 \pi}{p}} g( H(k)-E)_{\alpha,\alpha} \ \mathrm{d}k,
\end{equation}
    where $H(k)=[\mathcal{G}H\mathcal{G}^{-1}](k)$, $|\Gamma^{*}|=\frac{2\pi}{p}$ and $\alpha : R \alpha \in \mathcal{I}_{2,d}$ denotes the set of supercell orbital indices $\alpha \in \mathcal{A}^s$ corresponding to the $R \alpha \in \mathcal{I}_{2,d}$ in \eqref{equ:ldos}.
\end{theorem}

\begin{proof}
Since $g$ is a polynomial, it is suffice to prove that for the power function the identity holds. Suppose we use the $\psi^{\dagger}$ to denote the complex conjugate and transpose of the vector $\psi$. We use $\psi_{R\alpha}$ to denote a vector such that $\psi_{R\alpha}(R'\alpha')=\delta_{RR'}\delta_{\alpha,\alpha'}$. We then have, for any integer $m>0$,

\begin{equation}
\begin{split}
        \osum H^{m}(R\alpha,R\alpha) 
        &=\osum \psi^{\dagger}_{R\alpha} H^{m} \psi_{R\alpha} \\
        &= \osum \psi^{\dagger}_{R\alpha}  \mathcal{G}^{-1} (\mathcal{G}H\mathcal{G}^{-1})^{m} \mathcal{G}\psi_{R\alpha}.
\end{split}
\end{equation}
The second identity follows from  $\mathcal{G}^{-1}\mathcal{G}=I$, derived through definitions of the Bloch transform and its inverse for $\mathcal{R}^{s}$. 

\begin{lemma}
\label{lemma:inner}
For any $\psi, \phi \in \ell^{2}(\Omega)$, if we use $\psi^{\dagger}\phi$ to denote the inner product of the two vectors, then
\begin{equation}
    \psi^{\dagger}\phi =\sum_{\alpha \in \mathcal{A}^s}\int_{\Gamma^{*}} \widecheck{\psi}_{k}^{\dagger}(\alpha) \widecheck{\phi}_{k}(\alpha) \ \mathrm{d}k,
\end{equation}
    where $\widecheck{\psi}_{k}=[\mathcal{G}\psi](k)$ and $\widecheck{\phi}_{k}=[\mathcal{G}\phi](k)$ are defined as in \cref{eq:G}.
\end{lemma}
\begin{proof}
The inner product and inverse Bloch transform give that
\begin{equation}
    \begin{split}
        \psi^{\dagger}\phi &= \sum_{R\alpha} \psi^{\dagger}(R\alpha) \phi(R\alpha) \\
        &=\sum_{R\alpha} \frac{1}{|\Gamma^{*}|}\int_{\Gamma^{*}}e^{-ik'\cdot (R+\tau_{\alpha})}\widecheck{\psi}_{k'}^{\dagger}(\alpha)\ \mathrm{d}k' \int_{\Gamma^{*}} e^{ik\cdot (R+\tau_{\alpha})} \widecheck{\phi}_{k}(\alpha)\ \mathrm{d}k \\
        &=\frac{1}{|\Gamma^{*}|}\sum_{\alpha} \sum_{R} \int_{\Gamma^{*}} \int_{\Gamma^{*}} e^{i(k-k')\cdot R}e^{i(k-k')\cdot \tau_{\alpha}} \widecheck{\psi}_{k'}^{\dagger}(\alpha) \widecheck{\phi}_{k}(\alpha) \ \mathrm{d}k' \mathrm{d}k\\
        &=\sum_{\alpha} \int_{\Gamma^{*}} \int_{\Gamma^{*}} \delta(k-k') e^{i(k-k')\cdot \tau_{\alpha}}\widecheck{\psi}_{k'}^{\dagger}(\alpha) \widecheck{\phi}_{k}(\alpha) \ \mathrm{d}k' \mathrm{d}k\\
        &=\sum_{\alpha}\int_{\Gamma^{*}}   \widecheck{\psi}_{k}^{\dagger}(\alpha) \widecheck{\phi}_{k}(\alpha) \mathrm{d}k.
    \end{split}
\end{equation}

The fourth identity follows directly from Poisson summation formula and others follows from direct calculation.
\end{proof}

To continue derivation,
\begin{equation}
\label{equ:krepf}
    \begin{split}
        \osum \psi^{\dagger}_{R\alpha} \mathcal{G}^{-1} (\mathcal{G}H\mathcal{G}^{-1})^{m} \mathcal{G}\psi_{R\alpha}&= \osum \int_{\Gamma^{*}} \widecheck{\psi}^{\dagger}_{R\alpha,k} (\mathcal{G}H\mathcal{G}^{-1})^{m}\widecheck{\psi}_{R\alpha,k} \mathrm{d}k \\
        &=\frac{1}{|\Gamma^{*}|} \sum_{\alpha: R\alpha \in \mathcal{I}_{2,d}} \int_{\Gamma^{*}} H(k)_{\alpha,\alpha}^{m} \mathrm{d}k,
    \end{split}
\end{equation}
where $\widecheck{\psi}_{R\alpha,k}:=\mathcal{G}\psi_{R\alpha}(k)$. The first identity in \cref{equ:krepf} follows from \cref{lemma:inner} and the last identity follows from direct calculation.

\end{proof}

\section{Proof of \cref{thm:krepresent}}
\label{sec:krepresent2}
\begin{proof}

We define the truncated Hamiltonian $H_{\text{trunc}}$ with same intralayer coupling defined as in $\eqref{eq:hop0}$ and interlayer coupling as in $\eqref{eq:hop}$, except that there exists an positive integer $c>3$ such that 
\begin{equation}
    \label{equ:hc_cond2}
    H_{\text{trunc}}(R\alpha,R'\alpha')=0, \quad \text{if} \  |R-R'| > cp.
\end{equation}
Then the Bloch transformed truncated Hamiltonian $H_{k, \text{trunc}}$ is defined as $\mathcal{G}H_{\text{trunc}}\mathcal{G}^{-1}$. To estimate the numerical error between \eqref{eq:LDOS_momspace} and \eqref{eq:LDOS_momspace_trunc}, it suffices to estimate the following two terms,
\begin{equation}
\label{eq:B1plusB2}
\begin{split}
    & |\rho(d,E)-\varrho(d,E)| \\
    & 
    \leq \frac{1}{|\Gamma^*|}\left| \sum_{\alpha: R\alpha \in \mathcal{I}_{2,d}} \int_{0}^{\frac{2 \pi}{p}} g( H_{\theta,d}(k)-E)_{\alpha,\alpha} \text{d}k -\sum_{\alpha: R\alpha \in \mathcal{I}_{2,d}} \sum_{i=1}^{M} g(H_{\theta,d}(k_i)-E)_{\alpha,\alpha} \Delta k \right|\\
    &\quad \quad + \frac{1}{|\Gamma^*|}\left|\sum_{\alpha: R\alpha \in \mathcal{I}_{2,d}} \sum_{i=1}^{M} g(H_{\theta,d}(k_i)-E)_{\alpha,\alpha} \Delta k - \sum_{\mathcal{I}_{2,d}} \sum_{i=1}^{M} g(H_{k_i,\text{trunc}}-E)_{\alpha,\alpha} \Delta k\right|\\
    & =:B_1+B_2.
\end{split}
\end{equation}

\begin{lemma}
\label{lemma:H_k_periodic}
For any polynomial $g$ and any index $\alpha \in \mathcal{A}^{s}$, $g(H_{\theta,d}(k))_{\alpha,\alpha}$ is a periodic function in $k$ with period $|\Gamma^{*}|$. 
\end{lemma}

\begin{proof}
From the derivation of the Bloch transformed Hamiltonian in \cref{sec:Hk}, each component of the Hamiltonian has the format $\sum_{\Tilde{R}=R-R'} h(\Tilde{R}+\tau_{\alpha}-\tau_{\alpha'})e^{-ik[\Tilde{R}+\tau_{\alpha}-\tau_{\alpha'}]}$ where $h$ could be interlayer coupling or intralayer coupling. For any $m$-th power of Bloch transform of $H_{\theta,d}(k)$, the diagonal term $H_{\theta,d}^{m}(k)_{\alpha,\alpha}$ equals

\begin{equation}
    \sum_{\alpha_1, \alpha_2,\dots, \alpha_{m-1}} H_{\theta,d}(k)_{\alpha,\alpha_1}H_{\theta,d}(k)_{\alpha_1,\alpha_2} \dots H_{\theta,d}(k)_{\alpha_{m-1},\alpha},
\end{equation}
where $\alpha_i \in \mathcal{A}^s$ for $1\leq \alpha \leq m-1$. 
The phases depending on $k$ in addition to $e^{-ik\Tilde{R}}$ term are 
\begin{equation}
    e^{-ik[\tau_{\alpha}-\tau_{\alpha_1}]}e^{-ik[\tau_{\alpha_1}-\tau_{\alpha_2}]}\dots 
    e^{-ik[\tau_{\alpha_{m-1}}-\tau_{\alpha}]}=1.
\end{equation}
Therefore, the diagonal terms of Bloch transformed Hamiltonian are periodic with respect to the commensurate reciprocal lattice, i.e the period of diagonal term for Bloch transformed Hamiltonian is $|\Gamma^{*}|$.
\end{proof}

\begin{lemma}
\label{lemma:H_k_analytic}
    For any polynomial $g$ and any index $\alpha \in \mathcal{A}^{s}$,  $g(H_{\theta,d}(k))_{\alpha,\alpha}$ 
    extends to an analytic function in $k$ in a non-trivial strip containing the real axis.
\end{lemma}
\begin{proof}
    Since addition, multiplication of complex analytic function is still complex analytic function, it suffice to prove that each component of Bloch transformed Hamiltonian is complex analytic. To be specific, $\sum_{\Tilde{R}=R-R'} h_{inter}(\Tilde{R}+\tau_{\alpha}-\tau_{\alpha'})e^{-ik[\Tilde{R}+\tau_{\alpha}-\tau_{\alpha'}]}$ as a function of $k$ converges uniformly in the strip $|\text{Im}(k)|<\gamma$ due to exponential decay of the coefficient function $h_{inter}$ \cref{equ:h_cond}. Since similar remarks hold also for $h_{intra}$, and the uniform limit of analytic functions is analytic, the lemma is proved.
\end{proof}

We state a theorem on numerical quadrature error estimate as below (See \cite{trefethen2014exponentially} for more details and proofs).

\begin{theorem}[Trefethen]
\label{thm:quadrature}
    Suppose v is $T$-periodic and analytic and satisfies $|v(k)|\leq C_q$ in the strip $-a < \text{Im}(k) < a $ for some $a>0$. Then for any $M\geq 1$,
    \begin{equation}
        |\mathcal{I}_M-\mathcal{I}|\leq \frac{2TC_q}{e^{2\pi a M/T}-1}
    \end{equation}
where $\mathcal{I}=\int_{0}^{T}v(k)dk$ and $\mathcal{I}_M=\frac{T}{M}\sum_{i=1}^{M}  v(\frac{Ti}{M})$.
\end{theorem}

\cref{thm:quadrature} gives an upper bound for $B_1$ in \eqref{eq:B1plusB2} as  \cref{lemma:H_k_periodic} and \cref{lemma:H_k_analytic} guarantee that the condition of the theorem holds for $B_1$. The period $T=|\Gamma^{*}|=\frac{2\pi}{p}$ due to \cref{lemma:H_k_periodic} and the polynomials of Bloch transformed Hamiltonian is analytic due to \cref{lemma:H_k_analytic}. We let constant $a=\gamma/2$ in the case and the constants $C_q$ depend on polynomial $g$ and $a$. Therefore, we derive that

\begin{equation}
    B_1\leq \frac{4\pi C_q}{pe^{p\gamma M/2}-p} \cdot \frac{|\mathcal{I}_{2,d}|}{|\Gamma^{*}|}.
\end{equation}

 To estimate $B_2$, we bound the perturbation of Bloch transformed Hamiltonian $|H_{\theta,d}(k_i)-H_{k_i,\text{trunc}}|$. 
 The tight binding Hamiltonian with assumption \cref{equ:h_cond} can be proved to be bounded operator through standard results (see details in \cite{kong2023modeling} for example) and its Bloch transformed Hamiltonian is bounded as well. With Cauchy integral formulas, we have
\begin{equation}
\begin{split}
        g(H_{\theta,d}(k)-E)=\frac{1}{2\pi i} \int_{\Gamma_0} g(z-E)(z-H_{\theta,d}(k))^{-1} \ \mathrm{d}z\\
        g(H_{k,\text{trunc}}-E)=\frac{1}{2\pi i} \int_{\Gamma_0} g(z-E)(z-H_{k,\text{trunc}})^{-1} \ \mathrm{d}z
\end{split}
\end{equation}
where $\Gamma_0=\partial \mathcal{D} $ is a smooth contour and the enclosed domain $\mathcal{D}$ contains the spectrum of Bloch transformed Hamiltonian. The $\Gamma_0$ satisfies that for any $z\in \Gamma_0$, $|z-H|\geq d$ for some $d>0$. Due to the perturbation theory, the spectrum radius of $H_{\theta,d}(k)$ and $H_{k,\text{trunc}}$ are of similar size when $c$ is large enough. We are able to find one contour $\Gamma_0=\partial \mathcal{D}$ such that $\mathcal{D}$ contains spectrum of $H_{\theta,k}$ and $H_{k,\text{trunc}}$. To simplify notation, we use $g_E(x)$ to denote the function $g(x-E)$, and then 
\begin{equation}
    \begin{split}
    & \left |[g_E(H_{\theta,d}(k))-g_E(H_{k,\text{trunc}})]_{\alpha,\alpha} \right| \\
    &\leq 
        \left | \frac{1}{2\pi i}  \int_{\Gamma_0} g_E(z)(z-H_{\theta,d}(k))^{-1}(H_{\theta,d}(k)-H_{k,\text{trunc}})[(z-H_{k,\text{trunc}})^{-1}]_{\alpha,\alpha} \right |\\
        &\leq \frac{C_{\Gamma_0,g}}{2\pi} \sup_{z\in \Gamma_0}|[(z-H_{\theta,d}(k))^{-1}(H_{\theta,d}(k)-H_{k,\text{trunc}})(z-H_{k,\text{trunc}})^{-1}]_{\alpha,\alpha}|\\
        &\leq \frac{C_{\Gamma_0, g}C_{\Gamma_0}}{2d^2\pi} \sup_{\beta,\beta'\in \mathcal{A}^{s}} |[H_{\theta,d}(k)-H_{k,\text{trunc}}]_{\beta,\beta'}|,
    \end{split}
\end{equation}
where the last inequality uses the distance between $z$ and spectrum of Hamiltonian is at least $d$ and  $C_{\Gamma_0}$ depends on $\Gamma_0$ in \cref{equ:h_cond} and intralayer coupling. 
Note that we use $|\Gamma_0|$ to denote the length of the smooth contour and the constant $C_{\Gamma_0,g}:= |\Gamma_0| \|g\|_{L^{\infty}(\Gamma_0)}$. The last inequality is derived from the Combes-Thomas estimate~\cite{kong2023modeling,fischbacher2023lower} and Schur's test. The derivation of $H_{\theta,k}$ illustrates that
\begin{equation}
    \left |[H_{\theta,d}(k)]_{\alpha,\beta} \right| \leq C e^{-\gamma \min\{ |\tau_{\alpha}-\tau_{\beta}|, |\Gamma|-|\tau_{\alpha}-\tau_{\beta}| \} }.
\end{equation}
Through Combes-Thomas estimate~\cite{kong2023modeling,fischbacher2023lower}, we have
\begin{equation}
    \left |[(z-H_{\theta,d}(k))^{-1}]_{\alpha,\beta} \right| \leq \frac{C}{d} e^{-\gamma' \min\{ |\tau_{\alpha}-\tau_{\beta}|, |\Gamma|-|\tau_{\alpha}-\tau_{\beta}| \} },
\end{equation}
and
\begin{equation}
    \left |[(z-H_{k,trunc})^{-1}]_{\alpha,\beta} \right| \leq \frac{C}{d} e^{-\gamma' \min\{ |\tau_{\alpha}-\tau_{\beta}|, |\Gamma|-|\tau_{\alpha}-\tau_{\beta}| \} }
\end{equation}
where $|\Gamma|$ is the length of supercell. Indeed, we have that

\begin{equation}
    \begin{split}
        &\sup_{z\in \Gamma_0}|[(z-H_{\theta,d}(k))^{-1}(H_{\theta,d}(k)-H_{k,\text{trunc}})(z-H_{k,\text{trunc}})^{-1}]_{\alpha,\alpha}|\\
        =& \sup_{z\in \Gamma_0}|\sum_{\beta,\beta'}(z-H_{\theta,d}(k))_{\alpha,\beta}^{-1}(H_{\theta,d}(k)-H_{k,\text{trunc}})_{\beta,\beta'}(z-H_{k,\text{trunc}})^{-1}_{\beta',\alpha}| \\
         \leq& \sup_{z\in \Gamma_0}\sum_{\beta,\beta'}|(z-H_{\theta,d}(k))_{\alpha,\beta}^{-1}||(z-H_{k,\text{trunc}})^{-1}_{\beta',\alpha}| \sup_{\beta,\beta'} |(H_{\theta,d}(k)-H_{k,\text{trunc}})_{\beta,\beta'}|.
    \end{split}
\end{equation}

Note that
\begin{equation}
\begin{split}
    &\sum_{\beta,\beta'}|(z-H_{\theta,d}(k))_{\alpha,\beta}^{-1}||(z-H_{k,\text{trunc}})^{-1}_{\beta',\alpha}|\\
    =&\left( \sum_{\beta}|(z-H_{\theta,d}(k))_{\alpha,\beta}^{-1}| \right) \left( \sum_{\beta'} |(z-H_{k,\text{trunc}})^{-1}_{\beta',\alpha}|  \right).
\end{split}
\end{equation}
By the Combes-Thomas estimate, the sum can be upper bounded by a constant $C_{\Gamma_0}$.


Recall the definition of Bloch transformed Hamiltonian and assumption of interlayer coupling \cref{equ:h_cond}, we can derive that

\begin{equation}
\begin{split}
    |[H_{\theta,d}(k)-H_{k,\text{trunc}}]_{\beta,\beta'}| &= |\sum_{|R| > cp:R\in \mathcal{R}^{s}} h(R+\tau_{\beta}-\tau_{\beta'})e^{-ik[\Tilde{R}+\tau_{\beta}-\tau_{\beta'}]})|\\
    &\leq \frac{h_0e^{\gamma p}}{p} \int_{cp}^{\infty} e^{-\gamma r} \ \mathrm{d}r= \frac{h_0}{\gamma p} e^{\gamma p-\gamma cp}.
\end{split}
\end{equation}
Therefore, we derive
\begin{equation}
    B_2 \leq \frac{C_{\Gamma_0,g}h_0 N^2e^{-\gamma (c-1)p}}{2\gamma d^2 p \pi} \cdot \frac{|\mathcal{I}_{2,d}|}{|\Gamma^{*}|}.
\end{equation}
\end{proof}

\section{Derivation of Bloch Transformed Hamiltonian in coupled chain model}
\label{sec:Hk}
In this section, we derive the Bloch transformed Hamiltonian of commensurate coupled chain. We summarize this result in the following  \cref{thm:BlochH}. 

\begin{theorem}
\label{thm:BlochH}
    Suppose $\mathcal{G}$ and $\mathcal{G}^{-1}$ are Bloch transform and its inverse defined in \cref{eq:G} and \cref{equ:Ginverse} for supercell lattice $\mathcal{R}^{s}$ and indices set $\mathcal{A}^s$, $H$ is a tight binding Hamiltonian for a commensurate coupled chain defined as in \cref{equ:commensurateH}. Then for any $\psi \in \ell^{2}(\mathcal{R}^{s}\times \mathcal{A}^{s})$ and $k\in \Gamma^{*}$, there exists a $|\sci|\times|\sci|$ Hermitian matrix $H(k)$ denoted as $[\mathcal{G}H\mathcal{G}^{-1}](k)$ such that
    \begin{equation}
\label{equ:BtransH2}
    [\mathcal{G}H\psi](k)=H(k)[\mathcal{G}\psi] (k).
\end{equation}
Moreover, for $\alpha, \alpha' \in \mathcal{A}^{s}$ representing indices in the same layer
\begin{equation}
    H(k)_{\alpha,\alpha'}=\widecheck{h}_{\alpha,\alpha'}(k):=\sum_{\Tilde{R}\in \mathcal{A}^s} h_{intra}(\Tilde{R}+\tau_{\alpha}-\tau_{\alpha'})e^{-ik[\Tilde{R}+\tau_{\alpha}-\tau_{\alpha'}]},
\end{equation}
and for $\alpha,\beta \in \mathcal{A}^{s}$ representing indices in different layers,
\begin{equation}
    H(k)_{\alpha,\beta}=\widecheck{h}_{\alpha,\beta}(k):=\sum_{\Tilde{R}\in \mathcal{R}^s} h_{inter}(\Tilde{R}+\tau_{\alpha}-\tau_{\beta})e^{-ik[\Tilde{R}+\tau_{\alpha}-\tau_{\beta}]},
\end{equation}
where $h_{intra}$ and $h_{inter}$ are intralayer and interlayer coupling functions.
\end{theorem}

\begin{proof}

Recall the Bloch transform $\mathcal{G}:l^{2}(\mathcal{R}^{s}\times \mathcal{A}^{s}) \to C(\Gamma^{*}) \times \mathcal{A}^{s}$ is defined by
\begin{equation}
    [\mathcal{G}\psi]_{\alpha}(k)=\frac{1}{|\Gamma^{*}|^{\frac{1}{2}}}\sum_{R\in \scl} \psi(R\alpha) e^{-ik \cdot (R+\tau_{\alpha})},
\end{equation}
where $\widecheck{\psi}_{\alpha}(k)$ is used to denote $[\mathcal{G}\psi]_{\alpha}(k)$. To derive the Bloch transformed Hamiltonian, we are going to calculate
\begin{equation}
\label{equ:EBlochTrans}
    \mathcal{G}H\begin{bmatrix}
\psi_{1} \\
 \psi_{2}
\end{bmatrix}=\mathcal{G}\begin{bmatrix}
H_{11} & H_{12} \\
H_{21} & H_{22} 
\end{bmatrix}
\begin{bmatrix}
\psi_{1} \\
 \psi_{2}
\end{bmatrix}=\begin{bmatrix}
\mathcal{G}H_{11}\psi_1 + \mathcal{G} H_{12} \psi_2 \\
\mathcal{G}H_{21}\psi_1 + \mathcal{G} H_{22} \psi_2
\end{bmatrix}
\end{equation}
where $\psi_{i}$ is used to denote wavefunction in $i$-th layer and $H_{ij}$ denotes one block of Hamiltonian describing coupling between layer $i$ and layer $j$. We can derive the exact formula for $\mathcal{G}H_{11}\psi_1 + \mathcal{G} H_{12} \psi_2$ in \cref{equ:EBlochTrans}. In the derivation, we use $h_{intra}$ to denote intralayer coupling function and it is another representation of the model described in \cref{eq:hop0}. Indeed, $h_{intra}$ can be defined as,

\begin{equation}
    \label{equ:hintra}
    h_{intra}(R+\tau_{\alpha}-\tau_{\alpha'}) = \begin{cases}
       t_{\alpha,\alpha'},\, & R=a, \\
       t_{\alpha',\alpha},\, & R=-a, \\
       \epsilon_{\alpha,\alpha'}, \, &R=0,
       \\ 0, \, &\text{else},
    \end{cases}
\end{equation}
where $a$ is the length of one supercell. We start by calculating $\mathcal{G} H_{11} \psi_1(k)$ as
\begin{align*}    &[\mathcal{G}H_{11}\psi_1]_\alpha(k)\\
&= \frac{1}{|\Gamma^{*}|^{\frac{1}{2}}}\sum_{R\in \mathcal{R}^{s}} [ H_{11}\psi_{1}]({R\alpha})e^{-ik \cdot (R+\tau_{\alpha})}\\
    &=\frac{1}{|\Gamma^{*}|^{\frac{1}{2}}}\sum_{\substack{R \in \mathcal{R}^{s} \\ R'\in \mathcal{R}^{s}, \alpha'}} h_{intra}(R+\tau_{\alpha}-(R'+\tau_{\alpha'}))e^{-ik(R'+\tau_{\alpha'})} \psi_{1}(R'\alpha')e^{-ik[R+\tau_{\alpha}-R'-\tau_{\alpha'}]}\\
    &=\sum_{\Tilde{R}=R-R',\alpha'} h_{intra}(\Tilde{R}+\tau_{\alpha}-\tau_{\alpha'})e^{-ik[\Tilde{R}+\tau_{\alpha}-\tau_{\alpha'}]}\widecheck{\psi_{1,\alpha'}}(k)\\
    &=\sum_{\alpha'}\left(\sum_{\Tilde{R}=R-R'} h_{intra}(\Tilde{R}+\tau_{\alpha}-\tau_{\alpha'})e^{-ik[\Tilde{R}+\tau_{\alpha}-\tau_{\alpha'}]}\right)\widecheck{\psi_{1,\alpha'}}(k)\\
    &=:\sum_{\alpha'} \widecheck{h}_{\alpha,\alpha'}\widecheck{\psi_{1,\alpha'}}(k)
\end{align*}
where $\widecheck{\psi_{1,\alpha}}$ is the $\alpha$ component of the Bloch transform of $\psi$. Here $\widecheck{h}_{\alpha,\alpha'}$ is used to denote $\sum_{\Tilde{R}=R-R',\alpha'} h_{intra}(\Tilde{R}+\tau_{\alpha}-\tau_{\alpha'})$. The third identity holds due to the fact that the $\sum_{R \in \mathcal{R}^{s}, \alpha'} h_{intra}(R+\tau_{\alpha}-(R'+\tau_{\alpha'}))e^{-ik[R+\tau_{\alpha}-R'-\tau_{\alpha'}]}$ does not depend on $R'$ 
. 

For the off-diagonal block, we use $h_{inter}$ to denote the interlayer coupling function and can similarly derive that,
\begin{align*}
    &\mathcal{G}H_{12}\psi_{2}= \frac{1}{|\Gamma^{*}|^{\frac{1}{2}}}\sum_{R\in \mathcal{R}^{s}} [H_{12}\psi_{2}]({R\alpha})e^{-ik \cdot (R+\tau_{\alpha})}\\
    &=\frac{1}{|\Gamma^{*}|^{\frac{1}{2}}}\sum_{\substack{R \in \mathcal{R}^{s} \\ R'\in \mathcal{R}^{s}, \beta}}
h_{inter}(R+\tau_{\alpha}-(R'+\tau_{\beta}))e^{-ik(R'+\tau_{\beta})} \psi_{2}(R'\beta)e^{-ik[R+\tau_{\alpha}-R'-\tau_{\beta}]}\\
    &=\sum_{\Tilde{R}=R-R' \in \mathcal{R}^{s}, \beta} h_{inter}(\Tilde{R}+\tau_{\alpha}-\tau_{\beta})e^{-ik[\Tilde{R}+\tau_{\alpha}-\tau_{\beta}]}\widecheck{\psi_{2,\beta}}(k)\\
    &=\sum_{\beta}\left(\sum_{\Tilde{R}=R-R' \in \mathcal{R}^{s}} h_{inter}(\Tilde{R}+\tau_{\alpha}-\tau_{\beta})e^{-ik[\Tilde{R}+\tau_{\alpha}-\tau_{\beta}]}\right)\widecheck{\psi_{2,\beta}}(k)\\
    &=:\sum_{\beta} \widecheck{h}_{\alpha,\beta}\widecheck{\psi_{2,\beta}}(k).
\end{align*}

\end{proof}

\end{document}


\maketitle

\appendix
\subsection{Universal Approximation for the Operator with Encoder-Decoder neural network}

We now state the approximation theorem for CNN-based autoencoder-decoder structures. 

We use CNN's as the encoder and decoder structures. So far, no approximation theory agrees with the generic settings of CNN's. We first state various existing theorems on approximation theory of CNN's, and discuss the pros and cons of these theorems. 

\cite{Yarotsky2018CNNApprox} gives the analog of universal approximation theorem for CNN's with arbitrary input feature dimension and downsampling. 
\begin{definition}
\label{def:hbsps}
Define lattice $Z_L = \{k \in \mathbb{Z}^{\nu} | \|k\|_{\infty} \le L\}$. 
Define $V = L^2(\mathbb{R}^{\nu},\mathbb{R}^{d_V})$ as the Hilbert space of maps $\Phi : \mathbb{R}^{\nu} \rightarrow \mathbb{R}^{d_v}$ such that $\int |\Phi(\gamma)|^2 d\gamma < \infty$, and $V_{\lambda, \Lambda}$ as the Hilbert space of functions defined on a lattice $\lambda \mathbb{Z}^{\nu}$ with spatial cutoff $\Lambda$, i.e. 
$$
    V_{\lambda,\Lambda} = \{ \Phi \in L^2(\lambda  \mathbb{Z}^{\nu} , \mathbb{R}^{d_V}) | \Phi(\lambda k) = 0 \text{ if } k \not \in Z_{\lfloor \frac{\Lambda}{\lambda} \rfloor}\}.
$$

\end{definition}

\begin{definition}
\label{def:convnet1}
Let integer $L_{\text{rf}}$ be the CNN receptive field parameter, integer $s$ be the stride, and integers $d_1=d_v,d_2,\cdots,d_T,d_{T+1} = 1$ be the feature dimensions for each layer. Define a convnet with downsampling as follows: 
$$\hat f: V \xrightarrow{P_{\lambda, \lambda L_{1,T}}} V_{\lambda, \lambda L_{1,T}} (=W_1) \xrightarrow{\hat f_1} W_2 \xrightarrow{\hat f_2} \cdots \xrightarrow{\hat f_T} W_{T+1} (\cong \mathbb{R}) $$
with range parameters $L_{t,T} = 
\begin{cases}
L_{\text{rf}}(1+s+s^2+\cdots+s^{T-t-1}) ,&t < T \\
0, &t=T,T+1
\end{cases}
$ and layer operation formulas
$$\hat f_t(\Phi)_{\gamma n} = \sigma(\sum_{\theta \in Z_{L_{\text{rf}}}} \sum_{k=1}^{d_t} w_{n\theta k}^{(t)} \Phi_{s\gamma +\theta,k} + h_n^{(t)}), \gamma\in Z_{L_{t+1}}, n = 1,2,\cdots,d_{t+1}$$
$$\hat f_{T+1}(\Phi) = \sum_{k=1}^{d_{T}} w_{n k}^{(T)} \Phi_{k} + h_n^{(T)}$$
where the intermediate subspaces are defined as $W_t = L^2(s^{t-1}\lambda \mathbb{Z}_{L_t,T}, \mathbb{R}^{d_t})$. 
\end{definition}

More specifically, such a convnet is characterized by parameters $s, \lambda, L_{\text{rf}}, T, d_1,\cdots,d_T $ and coefficients $h_n^{(t)}$ and $w_{n\theta k}^{(t)}$. We have the following theorem: 

\begin{theorem}
\label{thm:CNNapprox}
If $f:V \rightarrow \mathbb{R}$ is continuous in the norm topology, then it is a limit point of convnets with downsampling, i.e. for any $s \le 2L_{\text{rf}}+1$, any compact set $K \in V$, any $\epsilon >0$, $\lambda_0>0$ and $\Lambda_0>0$ there exists a convnet with downsampling $\hat f$ with stride $s$, receptive field parameter $L_{\text{rf}}$, depth $T$ and spacing $\lambda \le \lambda_0$ such that $\lambda L_{1,T} \ge \Lambda_0$ and $\sup_{\Phi \in K} \|f(\Phi) - \hat f(\Phi) \| < \epsilon$.
\end{theorem}

We remark that $s \le 2L_{\text{rf}}+1$ is required so that the stride size is no larger than the size of filter windows, i.e. all information from previous layer is passed to the next layer. We further remark that \ref{thm:CNNapprox} easily extends to the case of multidimensional output, since it can be written as the stack of multiple convnets with one dimensional output. However, there is no estimate on the approximation accuracy with respect to depth, width and/or number of coefficients. 

Another theorem from \cite{Zhou2018CNNApprox} gives an estimate on the approximation accuracy for one-dimensional CNN's. 

\begin{theorem}
\label{thm:CNNapprox2}
Let $s$ be the filter length and $d$ be the input dimension. Assume $2 \le s \le d$. Let $f = F|_{\Omega}$ where $\Omega = [-1,1]^d$ and $F \in H^r(\mathbb{R}^d)$ with $r > 2+d/2$. Then there exists a CNN $\hat f$ with depth $J$ such that
$$\sup_{x \in \Omega} |f(x) - \hat f(x)| \le c\|F\|\sqrt{\log J}(1+J)^{\frac{1}{2}+\frac{1}{d}}$$
where $c$ is an absolute constant and $\|\cdot\|$ denotes the Sobolev norm on the Sobolev space $H^r(\mathbb{R}^d)$.
\end{theorem}

We remark that this theorem is only for one-dimensional CNN's, i.e. CNN's with one dimensional filters, with only one dimensional feature each layer. It is much further away from the real CNN people are using, with multiple channels each layer and often 2-dimensional filters. Moreover, the proof of \ref{thm:CNNapprox2} relies on results on ridge approximation that doesn't extend to the multidimensional feature case. Therefore in this work we only use \ref{thm:CNNapprox} as the CNN approximation theory, which is our basis for the autoencoder-decoder approximation theory. 

We now state the approximation theorem for CNN-based autoencoder-decoder structures. 

\begin{theorem}
Suppose there exist a homeomorphism f which is a mapping from $R^{n}$ to $R^{m}$ where $m<n$ and g is a continuous mapping from $R^{m+1}$ to $R^n$ such that g(f(x),0)=x. Then there exist an Encoder-decoder structure s.t
\begin{equation}
    |ED(x,t)-g(f(x),t)|\leq \epsilon
\end{equation}
\end{theorem}

To construct the proof, we need to use the neural network to approximate the three part of auto-encoder structure.
\begin{equation}
    f: x \to f(x)
\end{equation}
$f(x)$ is continuous mapping, thus {\color{red} by \ref{thm:CNNapprox}} there exists a {\color{red} C}NN, denoted as $E$ such that
\begin{equation}
    |f(x)-E(x)|\leq \epsilon
\end{equation}

A fully connected neural network is used between encoder and decoder. To reconcile theory and practice learning process, I is used to denote the fully connected neural network which is used to approximate the identity operator, i.e

\begin{equation}
    |I(u,t)-(u,t)|\leq \epsilon
\end{equation}

Similarly, there exist a {\color{red} CNN} decoder $D$ such that

\begin{equation}
    |g(u,t)-D(u,t)|\leq \epsilon
\end{equation}

In the first encoder step, $(x,t)$ is mapped into $(f(x),t)$. Therefore, we have
\begin{equation}
    |(f(x),t)-(E(x),t)|\leq \epsilon
\end{equation}

Estimate the error,

\begin{align*}
    &|D\circ I(E(x),t)-g(f(x),t)|\\&=|D\circ I (E(x),t)-g\circ I(E(x),t)+g\circ I(E(x),t)-g(E(x),t)+g(E(x),t)-g(f(x),t)|\\
    &\leq |D\circ I (E(x),t)-g\circ I (E(x),t)|+|g\circ I(E(x),t)-g(E(x),t)|+|g(E(x),t)-g(f(x),t)|\\
    &\leq \epsilon + C\epsilon+ C \epsilon
\end{align*}
Here C comes from the continuity of function g.

\section{Analysis and Numerical Scheme for SD-LDOS}
\label{sec:rdef}
We would like to give a derivation of numerical scheme for the stacking LDOS we defined in the paper in this section. Recall the definition of SD-LDOS in \ref{eq:LDOS},

\begin{equation}
\rho_{\theta,d}(E) = \sum_{o \in R_0} \sum_n |\psi^n_o|^2 \delta(E - E_n)
\label{eq:LDOS}
\end{equation}
where $o$ are the orbitals in the ``central'' unit-cell at position $R_0 = 0$, $n$ indexes the eigenstates of the system (summing over band indices and the BZ momentum), and $\{\psi^n_o$,$E_n\}$ are obtained from a Hamiltonian $H(\theta,d)$. The data generated in the paper for 1-D moir\'e structure are all rational numbers, indicating the coupled chains are commensurate structure. The unit cell length is set to be 1 and then there exist $p,q$ relative prime to each other satisfying,
\begin{equation}
    p=q(1-\theta)
\end{equation}
The supercell length is p and the moir\'e reciprocal space is [0,$\frac{2\pi}{p}$]. We would like to use supercell scheme to calculate the SD-LDOS which is corresponding to uniform sampling in the reciprocal space. Suppose we use M supercell, then there are $n=M(p+q)$ eigenstates in the system. Each state can also be indexed by a quasi-momentum vector $k_i$ and a band index $l$

Start with the definition of SD-LDOS,
\begin{align*}
            \rho_{\theta,d}(E) &= \sum_{o \in R_0} \sum_n |\psi^n_o|^2 \delta(E - E_n)\\
            &=\sum_{o \in R_0} \sum_{i,l} |\psi^{k_i,l}_o|^2 \delta(E - E_{k_i,l})
\end{align*}
We define the $\widecheck{\psi} ^l_{k_i}$ to be the eigenvector of momentum space Hamiltonian of the commensurate system and $\widecheck{\psi} ^l_{k_i,o}$ is the o th component in the vector.
Since numerically the eigenvector is always normalized to have $L_2$ norm to be 1, 

\begin{equation}
    M|\psi^{k_i,l}_o|^2=|\widecheck{\psi} ^l_{k_i,o}|^2
\end{equation}

Therefore the SD-LDOS is expressed as

\begin{equation}
    p \sum_{o \in R_0} \sum_{i,l} \frac{1}{2\pi}|\widecheck{\psi} ^l_{k_i,o}|^2 \delta(E - E_{k_i,l}) \Delta k
\end{equation}

where $\Delta=\frac{2\pi}{pM}$ is the sampling gap between two neighboring $k_i$. Replace the $\delta$ function with a normalized Gaussian g,
\begin{equation}
    g(x)=\frac{1}{\sqrt{2\pi}\sigma} e^{-\frac{x^2}{2\sigma^2}}.
\end{equation}

Then the numerical scheme is derived as

\begin{equation}
    p \sum_{o \in R_0} \sum_{i,l} \frac{1}{2\pi}|\widecheck{\psi} ^l_{k_i,o}|^2 g(E - E_{k_i,l}) \Delta k .
\end{equation}

Although a rigorous justification of convergence analysis is beyond the scope of the paper, we would also like to point out that when the uniform sampling is done infinitly dense,  

\begin{align}
    \lim_{\Delta k \to 0 } & p \sum_{o \in R_0} \sum_{i,l} \frac{1}{2\pi}|\widecheck{\psi} ^l_{k_i,o}|^2 \delta(E - E_{k_i,l}) \Delta k\\
    &=p \sum_{o \in R_0} [\sum_
    {l} \frac{1}{2\pi} \int_{0}^{\frac{2\pi}{p}}|\widecheck{\psi} ^l_{k,o}|^2 \delta(E-E_{k,l})dk]
\end{align}

\section{Inverse Problem from Piecewise constant SD-LDOS Scheme}
In this way, numerical SD-LDOS of aligned structure is computed by an explicit operation on real space Hamiltonian without doing integration. The numerical SD-LDOS can be defined as summation of delta function in real space Hamiltonian language, explained in \cref{sec:rdef}. But since $g$ here is polynomial or Gaussian, it is smooth with respect to E. Either we pick $ \widetilde \rho_{\theta,d}(E_{i})$ for each bin or we use the formula $ \frac{1}{\Delta E} \int_{E_{min}}^{E_{max}}\widetilde \rho_{\theta,d}(E)$, we will get an explicit formula for numerical local density of state.
\begin{equation}
    \widetilde \rho_{\theta,d,c}(E)=  \osum   \frac{1}{\Delta E}\sum_{l=1}^{2Mm} G_{E_1,E_2,    \epsilon}^{n}(\lambda_{k_i}^{(l)})   \Tilde{\psi}_{k_i}^{(l)} \otimes \tpsi_{k_i}^{(l),*} (Mm+o,Mm+o) h
\end{equation}
To be specific, if it is to be written in matrix operation form, we have
\begin{equation}
    \widetilde \rho_{\theta,d,c}(E)=    \frac{1}{\Delta E} \osum   G_{E_1,E_2, \epsilon}^{n}(\widetilde H_{d})(Mm+o,Mm+o) h
\end{equation}
where $\widetilde H_d$ is real space Hamiltonian with supercell scheme. If we choose to use average over bin of energy to be our final numerical local density of state, then 
\begin{align}
    \widetilde \rho_{\theta,d,c}(E)&=\frac{h}{\Delta E} \osum \int_{E_{min}}^{E_{max}}\widetilde \rho_{\theta,d}(E)=  \frac{h}{\Delta E} \osum \int_{E_{min}}^{E_{max}}  \sum_{l=1}^{2Mm} g_{E, \epsilon}^{n}(\lambda_{k_i}^{(l)})   \tpsi_{k_i}^{(l)} \otimes \tpsi_{k_i}^{(l),*} (Mm+o,Mm+o)  dE\\
    &=\frac{h}{\Delta E} \osum \sum_{l=1}^{2Mm}\int_{E_{min}}^{E_{max}}   g_{E, \epsilon}^{n}(\lambda_{k_i}^{(l)})   \tpsi_{k_i}^{(l)} \otimes \tpsi_{k_i}^{(l),*} (Mm+o,Mm+o)  dE\\
    &=\frac{h}{\Delta E}  \osum \sum_{l=1}^{2Mm} \int_{E_{min}}^{E_{max}}g_{E, \epsilon}^{n}(\lambda_{k_i}^{(l)})   \tpsi_{k_i}^{(l)} \otimes \tpsi_{k_i}^{(l),*} (Mm+o,Mm+o) dE\\
    &=\frac{h}{\Delta E} \osum G_{E_1,E_2, \epsilon}^{n}(\widetilde H_{d})(Mm+o,Mm+o) 
\end{align}

In particular, if we are using Gaussian to approximate delta function, the G function can be calculated explicitly. Since $g=\frac{1}{\sqrt{2\pi}\epsilon} e^{-\frac{(x-E)^2}{2\epsilon^2}}$, then $G_{E_1,E_2,\epsilon}=\int_{E_{1}}^{E_{2}}g_{E, \epsilon}^{n}(x)dE$. To be specific, we can compute this integral in a clear manner where $x$ is viewed as a parameter.
\begin{align}
    G_{E_1,E_2,\epsilon}&=\int_{E_{1}}^{E_{2}}g_{E, \epsilon}^{n}(x)dE
    =\int_{E_{1-x}}^{E_{2}-x}\sqrt{\frac{C}{\pi}} e^{-c^2 u^2}du=\frac{1}{2}[erf(\frac{E_2-x}{\sqrt{2} \epsilon})-erf(\frac{E_1-x}{\sqrt{2} \epsilon})]\\
    &=\frac{1}{2}[erf((\frac{\Bar{E}-x)+0.5 \Delta E}{\sqrt{2} \epsilon})-erf(\frac{(\Bar{E}-x)-0.5\Delta E}{\sqrt{2} \epsilon})]
\end{align}

Moreover, since the primative function of gaussian when $\epsilon$ goes to zero is approximating the Heaviside function and spectrum is uniformly bounded if range of material parameters are fixed.  For fixed $\epsilon$ and discretization step, the local density of state is explicitly given and if we use polynomial to approxiamte G function, the polynomial is invariant subject to shift. To be specific, 
\begin{equation}
    G_{E_1,E_2, \epsilon}^{n}(\widetilde H_{d})=\widetilde G^{n}(\widetilde H_{d}-\Bar{E})
\end{equation}
Here  $\widetilde G^{n}$ is a sequence of polynomials constructed in advance, only relevant to the way we approximate delta function and energy bin length.

\begin{theorem}
When $N_E>n$ there exist a one to one mapping between Hamiltonian on momentum space $H_k$ and numerical SD-LDOS, $\widetilde \rho_{d,c}(E)$ of aligned structure, computed from tight binding model with single parameter $t$. In particular, there exists a mapping $P_{h,N_D,N_E}$ such that
\begin{equation}
    P_{h,N_D,N_E}: \widetilde \rho_{d,c}(E) \to t
\end{equation}
 $h=\frac{1}{M}$ and M is the size of supercell and $N_D$ is resolution of local configuration and $N_E$ is the resolution of Energy.
\label{thm:inverse}
\end{theorem}
\begin{proof}
For a fixed numerical scheme which is using polynomials to approximate Gaussian which is approximating delta function in computation of local density of state, the number of highest degree of the polynomial is fixed, which is relevant to the order of numerical approximation. Either from momentum space Hamiltonian or real space Hamiltonian viewpoint, we can see that the local density of state is one diagonal term in matrix function. When local density of state is given with $N_E N_D$ resolutions, 
\begin{align*}
    \osum G^{n}(\widetilde H_{d}-\Bar{E})(Mm+o,Mm+o)=\osum \sum_{k=0}^{n} f_{k}(\Bar{E})H_d^{k}(Mm+o,Mm+o)
\end{align*}

is given at $E_j$ and $d_i$, $i=1\dots N_D, j=1\dots N_E$. For each $j$, and $i$ from $1$ to $n$.
\begin{equation}
    \sum_{k=0}^{n} f_{k}(E_i)H_{d_j}^{k}(Mm+1,Mm+1)
\end{equation}
the numerical SD-LDOS can be viewed as a linear system with unknowns \osum $H_{d_j}^{k}(Mm+o,Mm+o), k=1\dots n.$

\begin{lemma}
\label{lemma:inverse}
$G^{n}(\widetilde H_{d}-\Bar{E})=\sum_{k=0}^{n} f_{k}(E)H_d^{k}$ where $f_{k}(E)$ are linear independent functions. Moreover the matrix generated as following is non-singular, i.e, the matrix given by $(S)_{ij}=f_{j}(E_i)$ is non-singular matrix, i=1\dots n and j=1\dots n where n is the order of polynomial used to approximate erf function.  
\end{lemma}
\begin{proof}
Recall that the expansion of erf function is 
\begin{align}
    erf(x)&=\frac{2}{\sqrt{\pi}}\sum_{k=0}^{\infty} \frac{(-1)^{k} x^{2k+1}}{n!(2k+1)} , erf_n(x)=\frac{2}{\sqrt{\pi}}\sum_{k=0}^{n} \frac{(-1)^{k} x^{2k+1}}{n!(2k+1)}
\end{align}
where $erf_n$ represents n th order polynomial approximation. When n th order approximation is used, the $G_{E_1,E_2,\epsilon}$ can be expanded using the polynomial approximation. 
\begin{align}
    G_{E_1,E_2,\epsilon}(x)&=\frac{1}{2}[erf_n(\frac{(\Bar{E}-x)+0.5 \Delta E}{\sqrt{2} \epsilon})-erf_n(\frac{(\Bar{E}-x)-0.5\Delta E}{\sqrt{2} \epsilon})]\\
    &=\frac{1}{2}\frac{2}{\sqrt{\pi}}[\sum_{k=0}^{n} \frac{(-1)^{k} (\frac{(\Bar{E}-x)+0.5 \Delta E}{\sqrt{2} \epsilon})^{2k+1}}{n!(2k+1)}-\sum_{k=0}^{n} \frac{(-1)^{k} (\frac{(\Bar{E}-x)-0.5 \Delta E}{\sqrt{2} \epsilon})^{2k+1}}{n!(2k+1)} ]
\end{align}
Recall that $(a+b)^{2k+1}=\sum_{i=0}^{2k+1} C_{2k+1}^{i} a^{2k+1-i}b^{i}$ and $(a-b)^{2k+1}=\sum_{i=0}^{2k+1} C_{2k+1}^{i} (-1)^{i}a^{2k+1-i}b^{i}$ and $(a+b)^{2k+1}-(a-b)^{2k+1}$ can be computed with the expansion.
\begin{align}
    (a+b)^{2k+1}-(a-b)^{2k+1}=2\sum_{j=0}^{k} C_{2k+1}^{2j+1} a^{2k-2j}b^{2j+1}
\end{align}
Now apply the formula with $a=\frac{(\Bar{E}-x)}{\sqrt{2} \epsilon}, b=\frac{0.5\Delta E}{\sqrt{2} \epsilon}$ to $G_{E_1,E_2,\epsilon}(x)$ and replace $\Bar{E}$ with y to simplify the notation, 
\begin{align}
    G_{E_1,E_2,\epsilon}(x)=\sum_{j=0}^{n} c_{2j} (x-y)^{2j}
\end{align}
where $c_{2j}$ is constant relevant to $n, j, \epsilon$. Expand the G as a polynomial of x,
\begin{align}
        \sum_{j=0}^{n} c_{2j} (x-y)^{2j}&=\sum_{j=0}^{n} \sum_{i=0}^{2j} (-1)^{i }c_{2j}C_{2j}^{i}y^{2j-i}x^{i}\\&=\sum_{i=0}^{2n}\sum_{\frac{i}{2}\leq j \leq n} c_{2j}C_{2j}^{i}(-1)^{i} y^{2j-i}x^{i}
\end{align}

It can be seen that the coefficients of the polynomial for $x^{i}$ have the form, 
\begin{align}
    f_{i}(y)=\sum_{\frac{i}{2}\leq j \leq n} c_{2j}C_{2j}^{i}(-1)^{i} y^{2j-i}
\end{align}
The coefficients have special properties that the highest order term of $f_{2n-l}(y)$ is $y^{l}$ and there exists coefficients by mathematical induction such that
\begin{align}
f_{2n-l}(y)=a_{2n-l}y^{l}+\sum_{v=0}^{l-1}b_{2n-l,v}f_{2n-v}
\end{align}
where $a_{2n-l}$ is none zero. Therefore the linear independence of $f_0, f_{1},f_2\dots f_{2n}$ is proved due to linear independence of $y^{0}, y^{1}\dots y^{2n}$. Moreover, if we consider a matrix $(S)_{ij}=f_{j}(E_i)$ where $j=1\dots 2n$, $i=1\dots 2n$, it is a nonsingular matrix because after linear combination of columns of $S$ and rescaling, $S$ is turned into Vandermonde matrix $V$,  $(V)_{ij}=E_{i}^{j}$, which is nonsingular. Hence, $S$ is also nonsingular as well.
\end{proof}

Since the coefficient matrix is non-singular matrix, \osum $H_{d_j}^{k}(Mm+o,Mm+o)$ can be solved for each k. Therefore, for any $j\leq N_D$ and $k\leq n$, $\osum H_{d_j}^{k}(Mm+o,Mm+o)$ can be solved from linear system.

Then we are going to show how the $\osum H_{d_j}^{k}(Mm+1,Mm+1)$ can be used to solve for the parameter in tight binding model. Recall real space Hamiltonian is constructed with interlayer components and interlayer components. For the graphene-like bilayer system, two-band semiconductor and two-band metal, when intralayer coupling are parameterized by single variable.  Since $\osum H_{d_j}^{k}(Mm+o,Mm+o)$ is solved by solving the linear system, and in particular, $\osum H_{d_0}^{2}(Mm+o,Mm+o)$ is given. To be specific,
\begin{align}
       \osum H_{d_0}^{2}(Mm+o,Mm+o)&=\osum \sum_{j} H_{d_0}(Mm+o,j)^2  
      = at^2+b(d)
\end{align}

The coefficient $a$ in the expression is non zero since intralyer coupling is not all zero. Moreover, the parameters are assigned with sign when constructing real space hamiltonian. Thus the parameter $t$ can be solved directly. 

\end{proof}

\begin{remark}
 When the number of parameters increase, the proof could be tedious with current method. The problem is relevant to solving system of polynomials, which is a hard problem. Giving values of $\osum H_{d_j}^{k}(Mm+o,Mm+o)$ is equivalently to giving a sequence of polynomials equation from order $0$ to $k$ and the coefficients are functions of configuration space $d$. If we view the interlayer coupling term as one elements, then the polynomials are of integer coefficients.

In modern algebra, solving polynomial systems 
in real number is a hard problem. However, in our case, it could be possible. For example, when there are two parameters, x and y, 
\begin{align}
    H_{d_0}^{2}(Mm+1,Mm+1)&=c_1 x^2+c_2 y^2+a(d)\\
    H_{d}^{3}(Mm+1,Mm+1)&=a_1(d) x^3+a_2(d)x^2y+a_3(d)xy^2\\&+a_4(d) y^3+\dots + a_p(d)x+a_{p+1}(d)y+a_{p+2}(d) 
\end{align}
If needed, we can consider higher order result as well. Since the data is generated from two given parameters in advanced, we know that the system is solvable. Uniqueness is what left to be proved. However the coefficients functions of $d$ are highly relevant to the position of orbitals in each unit cell, which mean it is hard to give a proof for arbitrary bilayer structure.
\end{remark}

\begin{remark}
The interlayer coupling can be covered by taking advantage of varying local configuration of $d$. For fixed $E$, the numerical stacking dependent local density of state or numerical density of state gives an analytical function of d, which is relevant to the modelling of interlayer coupling. Assume we know the local geometry of untwisted bilayer material which is the case in the real DFT computation, the the local density of state is only a analytical function of d, $\nu$ and $R_{0}$ where  $\nu$ and $R_{0}$ are parameters for interlayer coupling. Moreover taking advantage of our \cref{lemma:inverse}, it is only left to show that given $H^{k}(i,i)$ for all $k$ and $i$ or  $\sum_{i} H^{k}(i,i)$ with different $d$, then $R_{0}$ and  $\nu$ can be shown to be uniquely decided. Since $H^{2}(i,i)$ is monotonic function of $\nu$ and $R_{0}$ respectively, if there exist that  $(R_0^{1},\nu_1)$ and $(R_0^{2},\nu_2)$, such that 
\begin{align}
       \osum H_{d}^{2}(Mm+o,Mm+o)&=\osum \sum_{j} H_{d}(Mm+o,j)^2  
      = at^2+b(d)
\end{align}
as a function of $d$, have same value for the two sets of parameters at $N_D$ $d$ values. However, two analytical functions can only have finite many intersections. Therefore, if when $N_D$ is large enough, the target functions still share the same value, then they must be the same function. However, the monotonic property imposes that $R_0^{1}>R_0^{2}$ and $\nu_1 < \nu_2$ and as a result the exponential functions in two functions are linear independent. This contradicts with the argument that the functions are same.

\end{remark}

\section{Remark on error}

There are three types of error to be measured when using NetA+NetB type method.
\begin{itemize}
    \item The neural net is approximation the operator $T_{\theta}P$ where P is inverse operator mapping local density of state to parameters in tight binding model. and $T_{\theta}$ is the mapping to twisted local density of state. The first type of error is error in the step of inverse operator. This is controlled by the continuity.
    \item The error from inverse operator is propagated by neural operator.
    \item The last part is the error due to discretization of twisted local density of state.
\end{itemize}

For one Network type method from end to end, then it is hard to use numerical analysis to give a rigorous error analysis because we have not explicit error bounds for the output of inverse problem. But this error can be bounded by the net A plus net B type error analysis.

\appendix
\section{old remarks}

\begin{remark}
Note that $\rho_{\theta,d}$ as a function of E has singularity, since $\frac{1}{\nabla_{k}\lambda_{k}^{(l)}}$ is infinity when E is local extreme. $\nabla_k \lambda_{k}^{(l)}$ could be zero on local minimum of the bands. The indicator of the emergence of flat bands with respect to LDOS is the intense singularity emerges as a period or a light spot instead of a straight line in the image of LDOS. 
\end{remark}

\begin{remark}
Although what have been defined is local density of state for bilayer material. The definition of local density of state for  multi-layer material can also be established. However, multi-layer material will need different topology of reciprocal space to describe LDOS. In the one dimensional case, the domain of the operator is in locally integrable function on $S^{1}$. If there are p layers, the local density of state is a locally integrable function on $T^{p-1}$ where $T$ is torus. The torus is mainly for showing the local geometry of multi-layer material.
\end{remark}

\begin{remark}
Note here that in real numerical computation, the eigenvector and eigenvalue are both real. However because of the commutation of real space Hamiltonian with permutation matrix as linear transform on complex space, they share the same complex eigenvectors. This is because of the permutation matrix P is normal and Real Hamiltonian is normal as well with $HP=PH$. But no matter use spectral decomposition on complex space or real space, it is the same symmetric matrix is being decomposed. For the analysis purpose, we use the eigenvectors of P as $\psi$.
\end{remark}

\begin{remark}
Another way to see this is from the commutation of two normal matrices. The momentum space Hamiltonian is hermitian and thus normal. The permutation matrix is circulant matrix and thus normal. To be specific, here we used the permutation matrix with resepct to symmetry with in the unit cell, i.e, they symmetry between top layer and bottom layer. It can be shown that the corresponding $\psi(j)=-\psi(j+m)$ for the eigenvector of permutation matrix. All process from complex eigenvectors to real eigenvectors are unitary transformation and does not affect the local density of state.
\end{remark}

\begin{remark}[repeated theorem]
\begin{theorem}
\label{thm:inverse2}
(change this part's notation)There exist a one to one mapping between Hamiltonian on momentum space $H_k$ and numerical local density of state $\widetilde \rho_{d,c}(E,A)$ where A is a finite set and each element in A represents one orbital in the unit cell, computed from material data for Tight binding model. In particular, there exists a mapping $P_{h,N_d,N_E}$ maps numerical local density of state to the parameters of Bilayer materials, i.e,
\begin{equation}
    P_{h,N_D,N_E}: \widetilde \rho_{\theta,d,c}(E) \to t
\end{equation}
Here $h=\frac{1}{M}$ and M is the size of supercell and $N_D$ is resolution of local configuration and $N_E$ is the resolution of Energy.
\end{theorem}
\begin{proof}
The proof is similar to the \cref{thm:inverse}. The only difference is we are solving equations
\begin{align}
    \sum_{j=Mm+1}^{Mm+m} H_{d_j}^{k}(Mm+j,Mm+j)
\end{align}
\end{proof}
\end{remark}

\begin{remark}[Inverse Theorem]
Using momentum space Hamiltonian is also tried in this case. Instead of using lemma for real space hamiltonian, \cref{lemma:inverse} is applied to momentum space hamiltonian. Then we derive $\sum_{i=0}^{M-1} H_{k_i,d}^{v}(m+1,m+1)$ for $v=0\dots n$ and fixed d. In particular, we have  $\frac{2\pi}{M}\sum_{i=0}^{M-1} H_{k_i,d}^{1}(m+1,m+1).$ However, from the construction of momentum space Hamiltonian, 
\begin{align}
    H_{k,d}(i,j)&=\sum_{R=-L,0,L} h(\Delta r_{ij}+R) e^{k \cdot (\Delta r_{ij}+R)}\\
    &=H_0(i,j)e^{k \cdot (\Delta r_{ij})}+H_x(i,j)e^{k \cdot (\Delta r_{ij}+L)}+H_{x}(j,i)e^{k \cdot (\Delta r_{ij}-L)}\\
    &=e^{k \cdot (\Delta r_{ij})}(H_0(i,j)+H_x(i,j)e^{k \cdot L}+H_{x}(j,i)e^{-k \cdot L})
\end{align}
Here L is the length of unit cell equals 1 in our computation. In the end, we derive the same equations like we did in real space Hamiltonian. From the viewpoints of solving system of polynomials, using real space hamiltonian or using momentum space hamiltonian do not matter.
\end{remark}

\begin{corollary}
\label{cor:noinverse}
Given a real space Hamiltonian H, and numerical density of state computed from polynomials of H, if H commute with a translation matrix, all the eigenvectors of H can not be recovered. 
\end{corollary}
\begin{proof}
Because H commute with translation matrix and thus they can be diagonalized at the same time which means that they have save complex eigenvectors. If we don't know the exact dependence of real space Hamiltonian on the parameters $H_0$ and $H_x$, it is impossible to recover the eigenvectors and as a result the real space hamiltonian. The reason is following. We actually know that the eigenvectors computed from numerical algorithm is real. However, we know the exact form of eigenvectors of P, which is translation matrix as P is a special form of ciculant matrix. The eigenvector is of the form $v_j=(1,w^{j},w^{2j},\dots,w^{(2mM-1)j})$ where j goes from 0 to $2mM-1$. The computed numerical eigenvectors are unitary transform of the complex eigenvectors.$(v'_j)=U(v_j)$. And the equation given by local density of state is invariant with respect to the any unitary matrix U.
\end{proof}

\begin{remark}
This corollary shows that to show the general case recovery theorem, we need to consider the delicate structure(block circulant property) of the real space Hamiltonian.
\end{remark}

\bibliographystyle{siamplain}
\bibliography{references}